%% file: main.tex
\documentclass[manuscript,acmsmall,screen]{acmart} 

\AtBeginDocument{%
	}

\setcopyright{acmcopyright}
\copyrightyear{2023}
\acmYear{2023}
\acmDOI{XXXXXXXX}

\acmJournal{TECS}
\acmVolume{0}
\acmNumber{0}
\acmArticle{0}
\acmMonth{0}

\include{preamble}
\begin{document}
	\title{Neural Abstraction-Based Controller Synthesis and Deployment} 

	\author{Rupak Majumdar}
	\email{rupak@mpi-sws.com}
	\orcid{0000-0003-2136-0542}
	\affiliation{%
		\institution{MPI-SWS}
		\city{Kaiserslautern}
		\country{Germany}
	}

	\author{Mahmoud Salamati}
	\email{msalamati@mpi-sws.com}
	\affiliation{%
		\institution{MPI-SWS}
		\city{Kaiserslautern}
		\country{Germany}
	}
	\orcid{0000-0003-3790-3935}
	
	\author{Sadegh Soudjani}
	\email{Sadegh.Soudjani@newcastle.ac.uk}
	\affiliation{%
		\institution{Newcastle University}
		\city{Newcastle}
		\country{United Kingdom}
	}
	\orcid{0000-0003-1922-6678}

	\renewcommand{\shortauthors}{Majumdar et al.}

	\begin{abstract}
		\input{abstract}
	\end{abstract}




\maketitle

\input{intro-new}

\input{prelims}
\input{problem}	
\input{TS_compression}
\input{cls_TS_compression}
\input{synthesis}

\input{deployment}
\input{experiments}
\input{conclusion}
\bibliographystyle{ACM-Reference-Format}
\bibliography{references}
\newpage


\end{document}

%% file: preamble.tex
\usepackage{amsmath}
\usepackage{amsfonts}
\usepackage{xcolor}
\usepackage{enumerate}
\usepackage{amsthm}
\usepackage{booktabs}
\usepackage{algpseudocode}
\usepackage{graphicx}
\usepackage{comment}
\usepackage{mathtools}
\usepackage{braket}
\usepackage{multirow}
\usepackage{caption}
\usepackage{subcaption}
\usepackage{cleveref}
\usepackage[linesnumbered, ruled]{algorithm2e}
\usepackage{algpseudocode}
\usepackage{stmaryrd}
\usepackage{dblfloatfix}
\usepackage{accents}
\usepackage{makecell}
\usepackage{hyperref}
\newtheorem{problem}{Problem}
\newtheorem{remark}{Remark}

\newcommand{\Sol}{\mathit{Sol}}

\newcommand{\Uh}{\widehat{U}}

\newcommand{\avoid}{\mathit{Avoid}}

\newcommand{\ball}{\Omega}

\def\reals{\mathbb{R}}
\def\ints{\mathbb{Z}}
\def\nats{\mathbb{N}}

\def\goal{\mathit{Goal}}

\def\init{\mathsf{in}}

\def\xb{\boldsymbol x}
\def\ub{\boldsymbol u}
\def\etab{\boldsymbol\eta}
\def\xbb{\bar{\boldsymbol x}}
\def\ubb{\bar{\boldsymbol u}}
\def\Xb{\bar{X}}
\def\Ub{\bar{U}}
\def\ds{\mathcal D}
\def\NN{\mathcal N}
\def\Xro{X_{\rho}}
\def\Xrob{\bar{X}_{\rho}}
\def\eFb{\boldsymbol e_F}
\def\eBb{\boldsymbol e_B}

\def\Post{{\mathrm{Post}}}
\def\Pre{{\mathrm{Pre}}}

\def\M{{\mathcal{M}}}
\def\Iu{{\mathcal I_u}}

\def\Ixi{{\mathcal I_{x,i}}}

\newcommand{\ubar}[1]{\underaccent{\bar}{#1}}

\newtheorem{example}{Example}

\usepackage[many]{tcolorbox}
\usetikzlibrary{calc}
\tcbuselibrary{skins}

\newtcolorbox{resp}[1][]{%
	enhanced jigsaw,%
	colback=gray!5!white,%
	colframe=gray!80!black,%
	size=small,%
	boxrule=1pt,%
	halign title=flush center,%
	coltitle=black,%
	breakable,%
	drop shadow=black!50!white,%
	attach boxed title to top left={xshift=1cm,yshift=-\tcboxedtitleheight/2,yshifttext=-\tcboxedtitleheight/2},%
	minipage boxed title=3cm,%
	boxed title style={%
		colback=white,%
		size=fbox,%
		boxrule=1pt,%
		boxsep=2pt,%
		underlay={%
			\coordinate (dotA) at ($(interior.west) + (-0.5pt,0)$);
			\coordinate (dotB) at ($(interior.east) + (0.5pt,0)$);
			\begin{scope}[gray!80!black]
				\fill (dotA) circle (2pt);
				\fill (dotB) circle (2pt);
			\end{scope}
		}%
	},%
	#1%
}

\makeatletter
\def\blfootnote{\xdef\@thefnmark{}\@footnotetext}
\makeatother

%% file: abstract.tex
Abstraction-based techniques are an attractive approach for synthesizing correct-by-construction controllers to satisfy high-level temporal requirements. 
A main bottleneck for successful application of these techniques is the memory requirement, both during controller synthesis (to store the abstract transition relation)
and in controller deployment (to store the control map).

We propose memory-efficient methods for mitigating the high memory demands of the abstraction-based techniques using \emph{neural network representations}.
To perform synthesis for reach-avoid specifications, we propose an on-the-fly algorithm that
relies on compressed neural network representations of the forward and backward dynamics of the system. 
In contrast to usual applications of neural representations, our technique maintains soundness of the end-to-end process.
To ensure this, we correct the output of the trained neural network such that the corrected output representations are sound with respect to the finite abstraction.
For deployment, we provide a novel training algorithm to find a neural network representation of the synthesized controller 
and experimentally show that the controller can be correctly represented as a combination of a neural network and 
a look-up table that requires a substantially smaller memory. 

We demonstrate experimentally that our approach significantly reduces the memory requirements of abstraction-based methods.
We compare the performance of our approach with the standard abstraction-based synthesis on several models.
For the selected benchmarks, 
our approach reduces the memory requirements respectively for the synthesis and deployment by a factor of $1.31\times 10^5$ and  $7.13\times 10^3$ on average, and up to $7.54\times 10^5$ and $3.18\times 10^4$. 
Although this reduction is at the cost of increased off-line computations to train the neural networks,
all the steps of our approach are parallelizable and can be implemented on machines with higher number of processing units to reduce the required computational time.
\blfootnote{This article appears as part of the ESWEEK-TECS special issue and was presented in the International Conference on Embedded Software (EMSOFT), 2023.} 

%% file: intro-new.tex
\section{Introduction}
Designing controllers for safety-critical systems with formal correctness guarantees has been studied extensively in the past two decades, 
with applications in robotics, power systems, and medical devices \cite{lee2016introduction,alur2015principles,kwiatkowska2005probabilistic}. 
Abstraction-based controller design (ABCD) has emerged as an approach that can algorithmically construct
a controller with formal correctness guarantees on systems with non-linear dynamics and 
bounded adversarial disturbances \cite{Tabuada2009book,belta2017formal,reissig2016feedback,rupak2020ouputfeedback,stanly2020}
and complex behavioral specifications.
ABCD schemes construct a \emph{finite abstraction} of a dynamical system that has continuous state and input spaces, and solve 
a two-player graph game on the abstraction.
When the abstraction is related to the original system through an appropriate behavioral relation (alternating bisimulation or feedback refinement \cite{reissig2016feedback}),
the winning strategy of the graph game can be refined to a controller for the original system. 
Finite abstractions can be computed analytically when the system dynamics are known and certain Lipschitz continuity properties hold. 
Even when the system dynamics are unknown, one can use data-driven methods to learn finite abstractions that are correct with respect to a 
given confidence \cite{Milad:2022,Arcak:2021, MAKDESI202149}. 

A main bottleneck of ABCD is the memory requirement, both in representing the finite abstract transition relation and in representing the controller.
First, the state and input spaces of the abstraction grow exponentially with the system and input dimensions, respectively, 
and the size of the abstract transition relation grows quadratically with the abstract states and linearly with the input states.
While symbolic encodings using BDDs can be used, in practice, the transition relation very quickly exceeds the available RAM.
Memory-efficient methods sometimes exploit the analytic description of the system dynamics or growth bounds \cite{Reissig:2022,Girard:2022,Rungger:2022}, but these techniques
are not applicable when the finite abstractions are learned directly from sampled system trajectories, 
or when a compact analytical expression of the growth bound is not available. 
Second, the winning strategy in the graph game is extracted as a look-up table mapping winning states to one or more available inputs.
Thus, the controller representation is also exponential in the system dynamics.
Such controllers cannot be deployed on memory-constrained embedded systems.

In this work, we address the memory bottleneck using approximate, compressed, representations of the transition relation and the controller map using neural networks.
We learn an approximate representation of the abstract transition relation as a neural network with a fixed architecture.
In contrast to the predominant use of neural networks to learn a generalization of an unknown function through sampling, we train the network on the entire data set
(the transition relation or the controller map) offline.
We store the transitions on disk, and train our networks in batch mode by bringing blocks of data into the RAM as needed. 
The trained network is small and fits into RAM.
Since the training of the network minimizes error but does not eliminate it, we apply a \emph{correction} to the output to ensure that the representation is 
\emph{sound} with respect to the original finite abstraction, i.e., every trajectory in the finite abstraction is preserved in the compressed representation.
We propose an on-the-fly synthesis approach that works directly on the corrected representation of the forward and backward dynamics of the system.
Although we present our results with respect to reach-avoid specifications, our approach can be generalized to other classes of properties and problems 
(e.g., linear temporal logic specifications \cite{baier2008principles}) in which the solution requires the computation of  the set of predecessors and successors in the underlying transition system.

Similarly, we store the winning strategy as a look-up table mapping states to sets of valid inputs on disk and propose a novel training algorithm to find a 
neural network representation of the synthesized controller.
The network is complemented with a look-up table that provides ``exceptions'' in which the network deviates from the winning strategy.
We experimentally demonstrate that a controller can be correctly represented as a combination of a neural network and a look-up table that requires a substantially smaller memory than
the original representation.

An important aspect of our approach is that, instead of using neural networks for learning an unknown data distribution, we train them over the entire data domain. 
Therefore, in contrast to many other applications wherein neural networks provide function representation and generalization over the unseen data, 
we are able to provide \emph{formal soundness guarantees} for the performance of the trained neural representations over the whole dataset.
 
Our compression scheme uses additional computation to learn a compressed representation and avoid the memory bottleneck.
In our implementation, the original relations are stored on the hard drive and
data batches are loaded sequentially into the RAM to perform the training. 
Hard drives generally have much higher memory sizes compared to the RAM, but reading data from the hard drive takes much longer. 
However, data access during training is predictable and we can perform prefetching to hide the latency. 
During the synthesis, the trained corrected neural representations fit into the RAM.
In contrast, a disk-based synthesis algorithm does not have predictable disk access patterns and is unworkable.
Similarly, the deployed controller only consists of the trained compact representation and (empirically) a small look-up table, 
which can be loaded into the RAM of the controlling chip for the real-time operation of the system.
 
We evaluate the performance of our approach on several examples of different difficulties and show that it is effective in reducing the memory requirements at 
both synthesis and deployment phases. 
For the selected benchmarks, our method reduces the space requirements of synthesis and deployment respectively by factors of $1.31\times 10^5$ and  $7.13\times 10^3$ in average, and up to $7.54\times 10^5$ and $3.18\times 10^4$,
 compared to the abstraction-based method that requires storing the full transition system. Moreover, we empirically show that, unlike other encodings, the memory requirement of our method is not affected by the system dimension on the considered benchmarks.

In summary, our main contributions are:
 \begin{itemize}
 	\item Proposing a novel and sound representation scheme for compressing finite transition systems using the expressive power of neural networks;
 	\item Proposing a novel on-the-fly controller synthesis method using the corrected neural network representations of forward and backward dynamics;
 	\item Proposing an efficient scheme for compressing the controller computed by abstraction-based synthesis methods;
 	\item Demonstrating significant reduction in the memory requirements by orders of magnitude through a set of standard benchmarks. \footnote{
 		Our implementations are available online at \url{https://github.com/msalamati/Neural-Representation}.
 	}
 \end{itemize}
The rest of this paper is organized as follows. After a brief discussion of related works, we give a high-level overview of our proposed approach in Subsec.~\ref{subsec:overview}.
The preliminaries and the problem statements are given in Sec.~\ref{sec:prelims}.
We provide the details of our synthesis and deployment algorithms in Sec.~\ref{sec:synthesis} and \ref{sec:deployment}, respectively.
In Sec~\ref{sec:experiments}, we provide experimental results of applying our approach to several examples.
We state the concluding remarks in Sec.~\ref{sec:conclusion}.

\subsection{Related Work}
%

\emph{Synthesis via reinforcement learning.}
The idea of using neural networks as function approximators to represent tabular data for synthesis purposes has been used in different fields such as reinforcement learning (RL) literature and aircraft collision avoidance system design. RL algorithms try to find an optimal control policy by iteratively guiding the interaction between the agent and the environment modeled as a Markov decision process \cite{Sutton1998}. When the space of the underlying model is finite and small, q-tables are used to represent the required value functions and the policy. When the space is large and possibly uncountable, such finite q-tables are replaced with neural networks as function approximators. Convergence guarantees that hold with the q-table representation \cite{Bertsekas99} are not valid for non-tabular setting \cite{Boyan1994,Bryant:1986,Tsitsiklis1997}. A similar behavior is observed in our setting: we lose the correctness guarantees in our approach without correcting the output of the neural network representations of the transition systems and the tabular controller.

\smallskip

\emph{Neural-aided controller synthesis.} Constructing neural network representations of the dynamics of the control system and using them for synthesis is studied in specific application domains including the design of unmanned airborne collision avoidance systems \cite{Julian:2016}. The central idea of \cite{Julian:2016} is to start from a large look-up table representing the dynamics, train a neural network on the look-up table, and use it in the dynamic programming for issuing horizontal and vertical advisories. Several techniques are used to reduce the storage requirement since the obtained score table---that is the table mapping every discrete state-input pair into the associated score---becomes huge in size (hundreds of gigabytes of floating numbers). Since simple techniques such as down sampling and block compression \cite{Kochenderfer:2013}, 
 are unable to achieve the required storage reduction, Julian et al. have shown that deep neural networks can successfully approximate the score table \cite{Julian:2019}. However, as in the RL controller synthesis, there is no guarantee that the control input computed using the neural representation matches the one computed using the original score table. In contrast, our corrected neural representations are guaranteed to produce formally correct controllers.

\emph{Reactive synthesis.} 
	Binary decision diagrams (BDDs) are used extensively in the reactive synthesis literature to represent the underlying transition systems 
	\cite{Roy:2011, Hsu:2018}. While BDDs are compact enough for low-order dynamical systems, recent synthesis tools such as SCOTS v2.0 
	\cite{Rungger2016scots} have already migrated into the non-BDD setting in order to avoid the large runtime overheads. In fact, motivated by reducing the required memory foot print, the current trend is to synthesize controllers in a non-BDD on the fly to eliminate the need for storing the transition system \cite{pFaces:2019,Khaled:2019,lavaei2020amytiss,Reissig:2022,Girard:2022,Rungger:2022}. These memory-efficient methods exploit the analytic description of the system dynamics or growth bounds. In contrast, our technique
is applicable also to the case where the finite abstractions are learned directly from the sampled system trajectories, 
i.e., when no compact analytical expression of the dynamics and growth bounds are available.


\emph{Verifying systems with neural controllers.}
An alternative approach developed for safety-critical systems is to use neural networks as a representation of the controller and learn the controller using techniques such as reinforcement learning and data-driven predictive control \cite{Dutta:2019,Tran:2019}.
In this approach, the controller synthesis stage does not provide any safety guarantee on the closed loop system, i.e., on the feedback connection of the neural controller and the physical system. Instead, the safety of the closed-loop system is verified a posteriori for the designed controller.
%
Ivanov et al. have considered dynamical systems with sigmoid-based neural network controllers, used the fact that
sigmoid is the solution to a quadratic differential equation to transform the composition of the system and the neural controller into an equivalent hybrid system, and studied reachability properties of the closed-loop system by utilizing existing reachability computation tools for hybrid systems \cite{Ivanov:2019}.
Huang et al. have considered dynamical systems with Lipschitz continuous neural controllers and used Bernstein polynomials for approximating the input-output model of the neural network \cite{Huang:2019}.
Development of formal verification ideas for closed-loop systems with neural controllers has led into emergence of dedicated tools such as NNV~\cite{Hoang-Dung:2020} and POLAR~\cite{Huang:2022}. While these methods provide guarantees on closed-loop control system with neural controllers, they can only consider \emph{finite horizon} specifications for a given set of initial states. In contrast, we consider controllers that are synthesized for \emph{infinite horizon} specifications.

\emph{Minimizing the memory foot print for symbolic controllers.} Girard et al. have proposed a method to reduce the memory needed to store safety controllers by determinizing them, i.e., choosing one control input per state such that an algebraic decision diagram (ADD) representing the control law is minimized \cite{GIRARD:2012,Girard:2022}. Zapreev et al. have provided two methods based on greedy algorithms and symbolic regression to reduce the redundancy existing in the controllers computed by the abstraction-based methods \cite{ZAPREEV:2018}. Both of the ADD scheme in \cite{GIRARD:2012,Girard:2022} and the BDD-based scheme in \cite{ZAPREEV:2018} have the capability to determinize the symbolic controller and reduce its memory foot print. However, the computed controller still suffers from the additional runtime overhead of the ADD/BDD encoding.
Further, as mentioned by the authors of \cite{ZAPREEV:2018}, their regression-based method is not able to represent the original controller with high accuracy. In contrast, our tool produces real-valued representations for symbolic controllers and can (additionally) be computed on top of the simplified version found by either of the methods proposed in \cite{Girard:2022,ZAPREEV:2018}.

\smallskip

\emph{Compressed representations for model predictive controllers (MPCs).}
Hertneck et al. have proposed a method to train an approximate neural controller representing the original robust (implicit) MPC satisfying the given specification \cite{Hertneck:2018}. While reducing the online computation time is the main motivation in implicit MPCs, minimizing the memory foot print is the main objective in explicit MPCs. Salamati et al. have proposed a method which is based on solving an optimization to compute a memory-optimized controller with mixed-precision coefficients used for specifying the required coefficients \cite{Salamati:2019}. Our method considers a different class of controllers that can fulfill infinite horizon temporal specifications.

 \subsection{Overview of the Proposed Approach}\label{subsec:overview}
 In this subsection, we provide a high-level description of our approach for both synthesis and deployment.

 \begin{figure*}[t]
	\vspace{-.2cm}
	\centering
	\includegraphics[scale=.3]{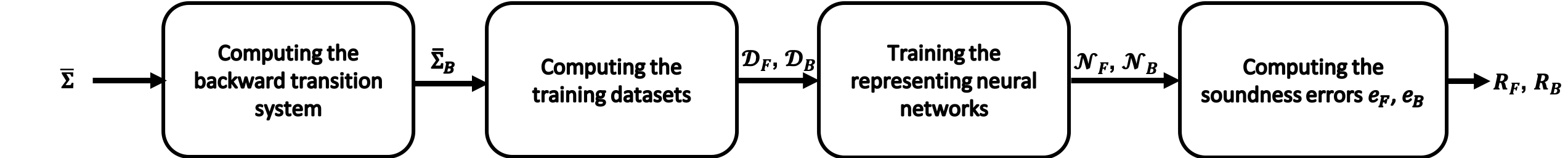}
	\caption{Graphical description of the proposed scheme for compressing finite abstractions}
	\label{fig:TS_compression_diagram}
\end{figure*}
 \smallskip
 \noindent\textbf{Corrected neural representations.}\
 Fig.~\ref{fig:TS_compression_diagram} 
 gives a pictorial description of the steps for computing a corrected neural network representation. Given a finite abstraction $\bar\Sigma$ that corresponds to the forward dynamics of the system and stored on the hard drive, we first compute the transition system $\bar\Sigma_B$ corresponding to the backward dynamics. Next, we extract the input-output training datasets  $\ds_F$ and $\ds_B$ respectively from the forward and backward systems, and store them on the hard drive. Each data point contains one state-input pair and the characterization of $\ell_\infty$ ball for the corresponding reachable set. We train two neural networks $\NN_F$ and $\NN_B$ such that they represent compressed input-output surrogates for the datasets $\ds_F$ and $\ds_B$, respectively. Finally, we compute the \emph{soundness errors} $\eFb$ and $\eBb$ which correspond to the difference between the output of $\NN_F$ and $\NN_B$ and the respective values in $\ds_F$ and $\ds_B$, calculated over all of the state-input pairs. We use the computed errors $\eFb$ and $\eBb$ in order to construct the corrected neural representations $R_F$ and $R_B$. We will get memory savings by using $R_F$ and $R_B$ instead of $\bar\Sigma$ and $\bar\Sigma_B$, respectively.

\begin{figure}[t]
	\centering
	\includegraphics[scale=.35]{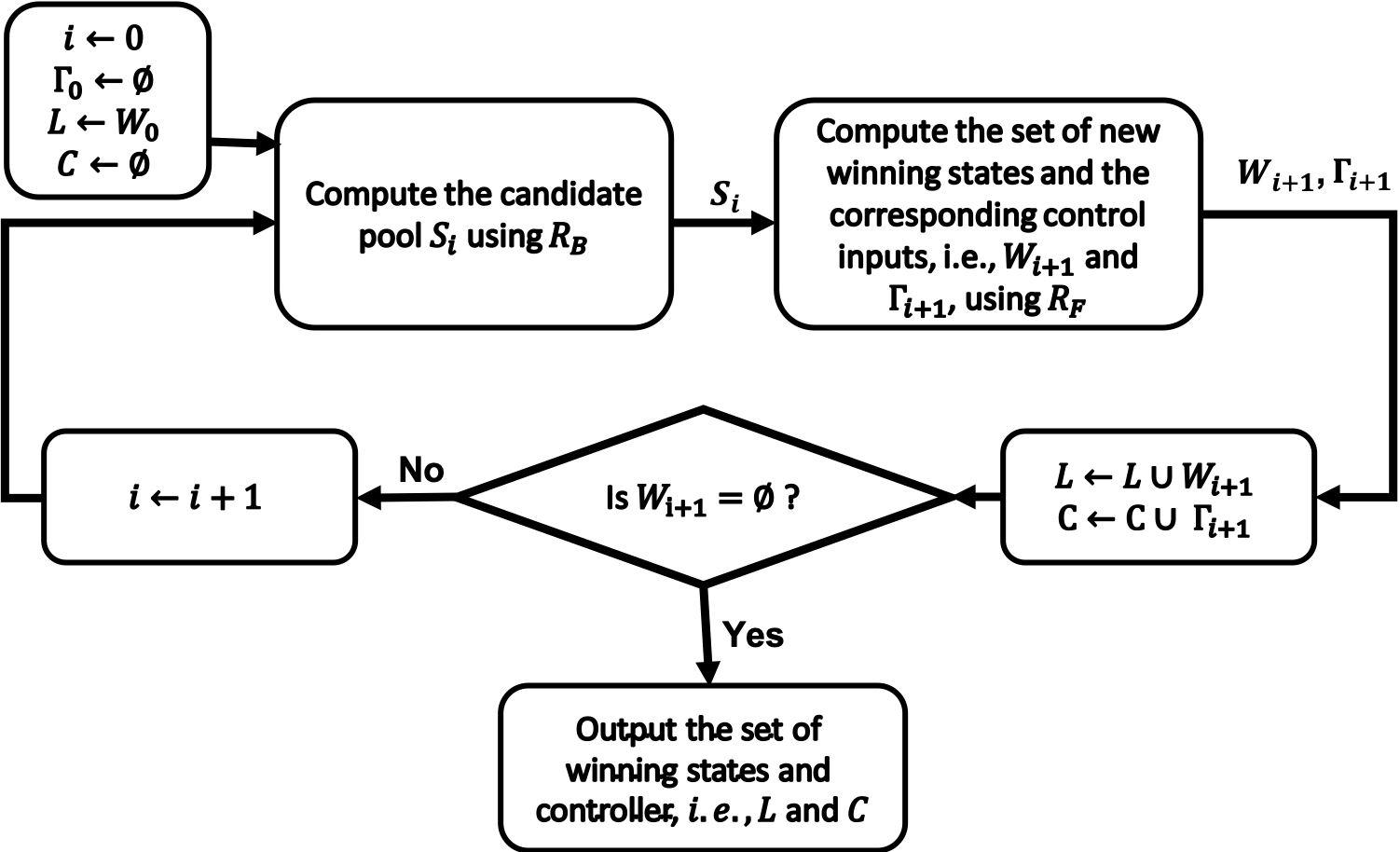}
	\caption{Graphical description of the proposed synthesis scheme}
	\vspace{-.5cm}
	\label{fig:synthesis_diagram}	
\end{figure}

\smallskip
 \noindent\textbf{Synthesis.}\
 Fig.~\ref{fig:synthesis_diagram} gives a pictorial description of our proposed synthesis algorithm for a reach-avoid specification with the target set  $\goal$ and obstacle set $\avoid$ as subsets of the state space. Let $W_0\subseteq \Xb$ represents a discrete under-approximation of the target set $\goal$. 
We initialize the winning set as $L=W_0$, the controller as $C=\emptyset$, and the set of state-input pairs that must be added to the controller as $\Gamma_{0}=\emptyset$. In each iteration, we compute the set of new states that belong to the winning set and update the controller accordingly, until no new state is added to $L$. To this end, we first use $R_B$ and its corresponding soundness error $e_B$ to compute a set of candidates $S_i$ out of which some belong to $L$ and it is guaranteed that there will be no winning state outside of $S_{i}$ in the $i^{th}$ iteration. 
We use $R_F$ and its corresponding soundness error $e_F$ to compute the set of new winning states $W_{i+1}\subseteq S_i$. We also compute the set of control inputs for every new winning state and compute the corresponding set of state-input pairs $\Gamma_{i+1}$ that must be added to the controller.
 Finally, if $W_{i}=\emptyset$, we terminate the computations as we already have computed the winning set $L$ and the controller $C$. Otherwise, we add the new winning set of states and state-input pairs, respectively, into the overall winning set ($L\leftarrow L\cup W_{i+1}$) and the controller ($C\leftarrow C\cup\Gamma_{i+1}$), and repeat the steps in the next iteration. 

\begin{figure*}[t]
	\centering
	\includegraphics[scale=.3]{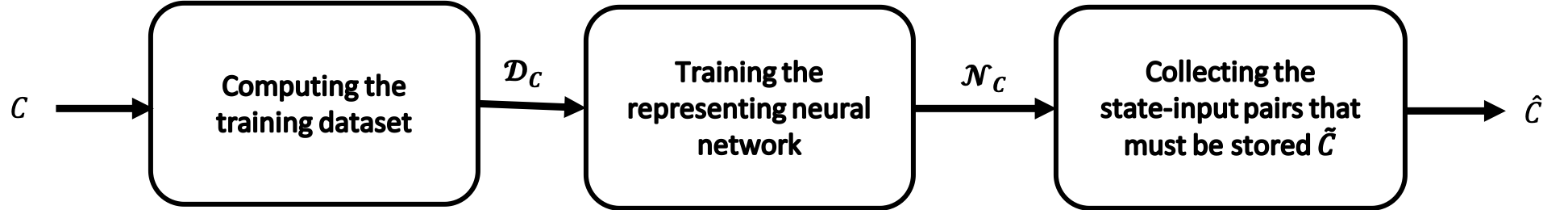}
	\caption{Graphical description of the proposed scheme for compressing  
		controllers 
	}
	\label{fig:cont_compression_diagram}
\end{figure*}

\smallskip
\noindent\textbf{Deployment.}\
Fig.~\ref{fig:cont_compression_diagram} 
shows our method for compressing controllers that are obtained from abstraction-based approaches. In the first step, we collect the training dataset $\ds_C$ and reformat it to become appropriate for our specific formulation of a classification problem. Each data point contains one state and an encoding of the corresponding set of control inputs.
We then train a neural network $\NN_C$ on the data with the loss function designed for this specific classification problem. Finally, we find all the states at which the output label generated by $\NN_C$ is invalid, and store the corresponding state-input pair in a look-up table, denoted by $\hat C$. We experimentally show that $\hat C$ only contains a very small portion of state-input pairs. 

%% file: prelims.tex

\section{Preliminaries}\label{sec:prelims}
\subsection{Notation}
We denote the set of integer numbers and natural numbers including zero by $\ints$ and $\nats$, respectively.
We use the notation $\reals$ and $\reals_{>0}$ to denote respectively the set of real numbers and the set of positive real numbers.
We use superscript $n>0$ with $\reals$ and $\reals_{>0}$ to denote the Cartesian product of $n$ copies of $\reals$ and $\reals_{>0}$ respectively. For a vector $\boldsymbol a \in \reals^d$, we denote its $i^{th}$ component, element-wise absolute value and $\ell_2$ norm by $\boldsymbol a(i)$, $|\boldsymbol a|$ and $\|\boldsymbol a\|$, respectively. For a pair of vectors $\boldsymbol a,\boldsymbol b\in\reals^d$, $\llbracket \boldsymbol{a},\boldsymbol{b}\rrbracket$ denotes the hyper-rectangular set $[\boldsymbol{a}(1),\boldsymbol{b}(1)]\times\dots\times[\boldsymbol{a}(d),\boldsymbol{b}(d)]$. Further, given $\boldsymbol{c}\in\reals^d$, $\boldsymbol{c}+\llbracket\boldsymbol{a},\boldsymbol{b}\rrbracket$ is another hyper-rectangular set which is shifted compared to $\llbracket\boldsymbol{a},\boldsymbol{b}\rrbracket$ to the extent determined by $\boldsymbol{c}$. Similarly, for a vector $\boldsymbol \eta\in\reals^d$ and a pair of vectors $\boldsymbol a,\boldsymbol b\in\reals^d$, for which $\boldsymbol a = \alpha \boldsymbol\eta$, $\alpha\in \ints$ and $\boldsymbol b=\beta\boldsymbol\eta$, $\beta\in\ints$, we define 
$\llbracket \boldsymbol{a},\boldsymbol{b}\rrbracket_{\boldsymbol\eta}=\prod_{i=1}^d A_i$,
where $A_i=\set{ \gamma\boldsymbol\eta(i)\mid\gamma\in \ints, \;\alpha\leq \gamma\leq \beta}.$ 
%
Given $\boldsymbol c\in \mathbb{R}^n$ and $\boldsymbol \varepsilon\in \mathbb{R}_{>0}^{n}$, the ball with center $c$ and radius $\boldsymbol\varepsilon$ in $\mathbb{R}^n$ is denoted by 
$\ball_{\boldsymbol\varepsilon}(\boldsymbol c)\coloneqq \set{\xb \in \mathbb{R}^n \mid  |\xb-\boldsymbol c|\leq \boldsymbol\varepsilon }$. For two integers $a,b\in\ints$, we define $[a;b]=\set{c\in \ints\mid a\leq c \leq b}$.

Let $A$ be a finite set of size $|A|$. The empty set is denoted by $\emptyset$. When $A$ inherits a coordinate structure, i.e., when its members are vectors on the Euclidean space, $A(i)$ denotes the projection of set $A$ onto its $i^{th}$ dimension. 
Further, we use the notation $A^\infty$ to denote the set of all finite and infinite sequences formed using the members of $A$. Our control tasks are defined using a subset of Linear Temporal Logic (LTL). 
In particular, we use the \emph{until} operator $\mathcal{U}$.
Let $p$ and $q$ be subsets of $\reals^n$ and $\rho=(\xb_0,\xb_1,\dots)$ be an infinite sequence of elements from $\reals^n$.  We write $\rho\models p\mathcal{U}q$ if there exists 
$i\in\mathbb{N}$ s.t. $\xb_i\in q$ and $\xb_j\in p$ for all $0\leq j<i$. 
For the detailed syntax and semantics of LTL, we refer to \cite{baier2008principles} and references therein.

\subsection{Control Systems}
We consider the class of continuous-state continuous-time control systems characterized by the tuple $\Sigma = (X, U, W, f)$, where
$X\subset \reals^n$ is the compact state space,
$U\subset\reals^m$ is the compact input space, and
$W\subset\reals^n$ is the disturbance space being a compact hyper-rectangular set of disturbances which is symmetric with respect to the origin (i.e., for every $\boldsymbol w\in W$ also it is the case that $-\boldsymbol w\in W$). 
The vector field $f: X \times U \rightarrow X$ is such that $f(\cdot, u)$ is locally Lipschitz for all $u\in U$.
The evolution of the state of $\Sigma$ is characterized by the differential inclusion
\begin{equation}
	\label{eq:ODE}
	\dot x(t)\in f(x(t),u(t))+W.
\end{equation}
Given a \emph{sampling time} $\tau>0$, an initial state $x_0\in X$, and a constant input $u\in U$, define the \emph{continuous-time trajectory} $\zeta_{x_0, u}$
of the system on the time interval $[0, \tau]$ as an absolutely continuous function $\zeta_{x_0,u}: [0, \tau] \rightarrow X$ such that $\zeta_{x_0,u}(0) = x_0$, and
$\zeta_{x_0,u}$ satisfies the differential inclusion $\dot{\zeta}_{x_0,u}(t) \in f(\zeta_{x_0,u}(t), u) + W$ for almost all $t\in [0,\tau]$.
Given $\tau$, $x_0$, and $u$, we define $\Sol(x_0, u, \tau)$ as the set of all $x\in X$ such that there is a continuous-time trajectory
$\zeta_{x_0,u}$ with $\zeta(\tau) = x$.
A sequence $x_0,x_1,x_2, \ldots$ is a \emph{time-sampled trajectory} for a continuous control system if $x_0\in X$ and for each $i\geq 0$, we have
$x_{i+1} \in \Sol(x_i, u_i, \tau)$ for some $u_i \in U$.

\subsection{Finite Abstractions}
In order to satisfy a temporal specification on the trajectories of the system, it is generally needed to over-approximate the dynamics of the system with a finite discrete-time model. 
Let $\Xb\subset X$ and $\Ub\subset U$ be the finite sets of states and inputs, computed by (uniformly) quantizing the compact state and input spaces $X$ and $U$ using the  rectangular discretization partitions of size $\etab_x\in\reals^n_{>0}$ and $\etab_u\in\reals^m_{>0}$, respectively. 
A \emph{finite abstraction} associated with the dynamics in Eq.~\eqref{eq:ODE} 
is characterized by the tuple $ \bar\Sigma\colon(\Xb, \Ub, T_F)$, where $T_F\subseteq\Xb\times \Ub \times \Xb$ denotes the system's \emph{forward-in-time transition system}. The transition system $T_F$ is defined such that
\begin{align*}
&(\xbb,\ubb,\xbb')\in T_F\Leftrightarrow\exists (\xb,\ub,\xb')\in \ball_{\frac{\etab_x}{2}}(\xbb)\times \ball_{\frac{\etab_u}{2}}(\ubb)\times \ball_{\frac{\etab_x}{2}}(\xbb')\;\text{s.t.}\;
\xb'\in\Sol(\xb,\ub,\tau).
\end{align*}
When the dynamics in Eq.~\eqref{eq:ODE} are known and satisfy the required Lipschitz continuity condition, the finite abstraction can be constructed using the method proposed in \cite{reissig2016feedback}. For systems with unknown dynamics, data-driven schemes for learning finite abstractions can be employed \cite{Milad:2022,Arcak:2021, MAKDESI202149}.
By abusing the notation, we denote the \emph{reachable set} for a state-input pair $(\xbb,\ubb)\in\Xb\times\Ub$ by $T_F(\xbb,\ubb)=\set{\xbb'\in\Xb\mid \xbb'\in \Sol(\xbb,\ubb,\tau)}$.  
We assume that the reachable sets take hyper-rectangular form, meaning that for every $\xbb\in \Xb$, $\ubb\in\Ub$ the corresponding reachable set $H=T_F(\xbb,\ubb)$ can be rewritten as $H=\prod_{i=1}^n H(i)$, where $H(i)$ corresponds to the projection of the set $H$ onto its $i^{th}$ coordinate. Otherwise, in case that $H$ is not hyper-rectangular, it is over-approximated by $\prod_{i=1}^n H(i)$. 
Note that $\bar\Sigma$ can in general correspond to a \emph{non-deterministic} control system, i.e., $|T_F(\xbb,\ubb))|>1$ for some $\xbb\in\Xb, \ubb\in\Ub$. 

Given $\bar\Sigma$, one can easily compute the characterization of the \emph{backward-in-time} dynamics as
 \begin{equation}\label{eq:BW_sys}\bar\Sigma_B=(\Xb,\Ub,T_B),\; T_B=\set{(\xbb,\ubb,\xbb')\in\Xb\times\Ub\times\Xb\mid(\xbb',\ubb,\xbb)\in T_F}.
 \end{equation} 
A \emph{trajectory} of $\bar\Sigma$ is a finite or infinite sequence $\xb_0, \xb_1,\xb_2, \ldots \in \Xb^\infty$,
such that for each $i\geq 0$, there is a control input $\ubb_i\in \Ub$ such that
$(\xb_i,\ubb_i,\xb_{i+1}) \in T_F$. 
The operator $\Pre(\cdot)$ acting on sets $P\subseteq\Xb$ is defined as
\begin{align*}
	&\Pre(P)=\{\xbb\in \Xb \mid \exists \ubb\in\Ub\;s.t.\; T_F(\xbb,\ubb)\subseteq P\}.
\end{align*}
Finally, to compute an over-approximating set of the discrete states that have overlap with a hyper rectangular set $\llbracket\boldsymbol x_{lb}, \boldsymbol x_{ub}\rrbracket$, we define the (over-approximating) quantization mapping as
\begin{equation*}
	\bar K(\boldsymbol x_{lb}, \boldsymbol x_{ub}) = \set{\bar{\boldsymbol x}'\in \bar X\mid \llbracket\bar{\boldsymbol x}'-\boldsymbol{\eta_x}/2, \bar{\boldsymbol x}'+\boldsymbol{\eta_x}/2\rrbracket\cap \llbracket\boldsymbol x_{lb}, \boldsymbol x_{ub}\rrbracket\neq \emptyset}.
\end{equation*}
Similarly, the under-approximating quantization mapping is defined as
\begin{equation*}
	\ubar K(\boldsymbol x_{lb}, \boldsymbol x_{ub}) = \set{\bar{\boldsymbol x}'\in \bar X\mid \llbracket\bar{\boldsymbol x}'-\boldsymbol{\eta_x}/2, \bar{\boldsymbol x}'+\boldsymbol{\eta_x}/2\rrbracket\subseteq \llbracket\boldsymbol x_{lb}, \boldsymbol x_{ub}\rrbracket}.
\end{equation*}

\subsection{Controllers}
For a finite abstraction $\bar\Sigma=(\Xb,\Ub,T_F)$, a feedback controller is denoted by $C\subseteq\Xb\times\Ub$. The set of valid control inputs at every state $\xbb\in\Xb$ is defined as $C(\xbb)\coloneq \set{\ubb\in\Ub\mid(\xbb,\ubb)\in C}$.
We denote the feedback composition of $\bar\Sigma$ with $C$ as $C\parallel\bar\Sigma$. For an initial state $\xbb^\ast\in\Xb$, the set of trajectories of $C\parallel\bar\Sigma$ having length $k\in\nats$ is the set of sequences $\xbb_0,\xbb_1,\xbb_2,\dots,\xbb_{k-1}$, s.t. $\xbb_0=\xbb^\ast$, $\xbb_{i+1}\in T_F(\xbb_i,\ubb_i)$ and $\ubb_i\in C(\xbb_i)$ for $i\in[0;k-2]$. 
\subsection{Neural Networks}
A \emph{neural network} $\NN(\boldsymbol\theta,\cdot)\colon \reals^d\rightarrow\reals^q$ of depth $v\in\nats$ is a parameterized function which transforms an input vector $\boldsymbol a\in\reals^d$ into an output vector $\boldsymbol b\in\reals^q$, and is constructed by the forward combination of $v$ functions as follows:
$$
\NN(\boldsymbol\theta,\boldsymbol a)=G_v(\boldsymbol\theta_v,G_{v-1}(\boldsymbol\theta_{v-1},\ldots,G_2(\boldsymbol\theta_2,G_1(\boldsymbol\theta_1,\boldsymbol a)))),
$$
where 
$\boldsymbol\theta=(\boldsymbol\theta_1,\dots,\boldsymbol\theta_v)$ and
$G_i(\boldsymbol{\theta}_i,\cdot)\colon \reals^{p_{i-1}}\rightarrow \reals^{p_{i}}$ denotes the $i^{th}$ \emph{layer} of $\NN$
parameterized by $\boldsymbol\theta_i$ with $p_0=d$, $p_i\in \nats$ for $i\in[1;v]$ and $p_v=q$. The $i^{th}$ layer of the network, $i\in[1;v]$, takes an input vector in $\reals^{p_{i-1}}$ and transforms it into
an output representation in $\reals^{p_i}$ depending on the value of parameter vector
$\boldsymbol \theta_i$ and type of the used \emph{activation function} in $G_i$. During the \emph{training phase} of the network, the set of parameters $\boldsymbol{\theta}$ is learned over the \emph{training set} which consists of a number of input-output pairs $\{(\boldsymbol a_k,\boldsymbol b_k)\mid k=1,2,\ldots,N\}$, 
in order to achieve the highest performance with respect to an appropriate metric such as mean squared error. For a trained neural network, we drop its dependence on the parameters $\boldsymbol\theta$. In this paper, we characterize a neural network of depth $v$ using its corresponding list of layer sizes, i.e., $(p_1,p_2,\ldots,p_v)$, and the type of the activation function used, e.g., hyperbolic tangent, Rectified Linear Unit (ReLU), etc.

Neural networks can be used for both \emph{regression} and \emph{classification} tasks. In a regression task, the goal is to predict a numerical value given an input, whereas, a classification task requires predicting the correct class label for a given input. In order to measure performance of the trained neural network, we consider \emph{prediction error}. Note that prediction error is different from the metrics such as \emph{mean squared error} (MSE) which are used during the training phase for defining the objective function for the training. The prediction error for regression and classification tasks is defined differently. For \emph{our regression tasks}, we define the prediction error for a trained neural network $\NN$ over a training set $\{(\boldsymbol a_k,\boldsymbol b_k)\mid k=1,2,\ldots,N\}$ as
$$\boldsymbol e=\max_{k\in[1;N]}|\NN(\boldsymbol a_k)-\boldsymbol b_k|.$$
In this paper, we consider the classification tasks wherein there may exist more than one valid class label for each input. Therefore, the training set would be of the form $\{(\boldsymbol a_k,\boldsymbol b_k)\mid k=1,2,\ldots,N\}$, where $b_k\in\{0,1\}^q$ 
and $\boldsymbol b_k(i)=1$ iff $i\in[1;q]$ corresponds to a valid label at $\boldsymbol a_k$. 
Since the number of valid labels for each input can be different, we define the prediction error of a trained classifier $\NN$ in the following way:
$$err=\frac{|\set{k\in[1;N]\mid \boldsymbol b_k(i)=0\;\text{with }   i=argmax(\NN(\boldsymbol a_k))}|}{N}.$$
For a given neural network $\NN$ with the training set $\{(\boldsymbol a_k,\boldsymbol b_k)\mid k=1,2,\ldots,N\}$, we define the \emph{continuity index} as 
\begin{equation}\label{eq:cont_idx}
\alpha_\NN =  \max_{1\leq i,j\leq N,\;i\neq j}\frac{\|\NN(\boldsymbol a_i) -\NN( \boldsymbol a_j)\|}{\|a_i-a_j\|}.
\end{equation}

%% file: problem.tex

\subsection{Problem Statement}\label{subsec:problem}

We now consider the controller synthesis problem for finite abstractions 
w.r.t.\ a reach-avoid specification. 
Let $\goal, \avoid \subseteq X, \goal\cap\avoid =\emptyset$ be the set of states representing the target and unsafe spaces, respectively. 
The \emph{winning domain} for the finite abstraction $\bar\Sigma=(\Xb,\Ub,T_F)$ is the set of states $\xbb^\ast\in\Xb$ such that there exists a feedback controller $C$ such that all trajectories of $C\parallel\bar \Sigma$, which are started at $\xbb^\ast$, satisfy the given specification $\Phi$. $\xbb_0=\xbb^\ast,\xbb_1,\xbb_2,\dots\models\Phi$. 
 The aim is to find the set of the winning states $L$ together with a feedback controller $C$ such that $C\parallel \bar\Sigma$ satisfies the reach-avoid specification $\Phi $. To compute the winning domain and the controller, one can use the methods from reactive synthesis. For many of interesting control systems, size of $T_F$ in the finite abstraction becomes huge. This restricts the application of reactive-synthesis-based methods for computing the controller. Therefore, we are looking for a method which uses compressed surrogates of $T_F$ to save memory. In particular, we want to train two \emph{corrected neural} surrogates, i.e., neural network representations whose output is corrected to maintain the soundness property: $R_F$ for the forward-in-time dynamics and $R_B$ for the backward-in-time dynamics. 

%
\begin{resp}
	\begin{problem}
		\label{prob:discrete}
		\noindent\textbf{Inputs:} Finite abstraction $\bar\Sigma=(\Xb,\Ub,T_F)$, and specification $\Phi=\lnot\avoid\,\mathcal{U}\,\goal$.
		
		
		\noindent\textbf{Outputs:} Corrected neural representations $R_F$ and $R_B$,  
		winning domain $L$ and  a feedback controller $C$ for $\bar\Sigma$ such that $C\parallel\bar\Sigma$ realizes $\Phi$. 
	\end{problem}
\end{resp}

It is important to notice that any solution for this problem is required to provide a \emph{formal guarantee} on the satisfaction of $\Phi$, 
i.e., the reach-avoid specification $\Phi$ \emph{must} be satisfied under any disturbance affecting the control systems.

Let $C\in \Xb\times \Ub$ be the computed controller for the abstraction $\bar\Sigma$ such that $C\parallel\bar\Sigma$ realizes a given specification $\Phi$. The size of this controller can be large due to the large number of discrete state and inputs. For deployment purposes, we would like to compute a corrected neural controller $\hat C\coloneqq \Xb\rightarrow\Ub$ s.t. $\hat C\parallel \bar\Sigma$ realizes $\Phi$. 
\begin{resp}
	\begin{problem}
		\label{prob:controller_compression}
		\noindent\textbf{Inputs:} Controller $C$ computed for the discrete control system $\bar\Sigma$, and specification $\Phi$ s.t. $C\parallel\bar\Sigma$ realizes $\Phi$.
		
		\noindent\textbf{Outputs:} A corrected neural controller $\hat C$  
		 such that $\hat C\parallel\bar\Sigma$ realizes $\Phi$. 
	\end{problem}
\end{resp}
\vspace{-.2cm}

%% file: TS_compression.tex

\section{Synthesis}\label{sec:synthesis}
One approach to formally synthesize controllers for a given specification is to store the transition system corresponding to quantization of the state and input spaces, and to use the methods from reactive synthesis to design a controller. However, the memory required to store these transition systems increases exponentially with the number of state variables, which causes a memory blow-up for many real-world systems.
In this section, we propose our memory-efficient algorithm for synthesizing controllers to satisfy reach-avoid specifications for finite abstractions and reach-avoid specifications. Our method requires computation of corrected neural representations for the finite abstraction. Computation of these representations is discussed in Sec.~\ref{subsec:TS_compression_regression}. Later, in Sec.~\ref{subsec:on-the-fly_synthesis}, we show how our synthesis method makes use of the computed representations.
\subsection{Corrected Neural Representations for Finite Abstractions}\label{subsec:TS_compression_regression}
Let $\bar\Sigma=(\Xb,\Ub,T_F)$ be a finite abstraction. In this section, we show that 
$T_F$ can be approximated by some \emph{generator functions}. In particular, we show how to compute generator functions $R_F\colon \Xb\times\Ub\rightarrow\reals^n\times\reals^n_{\geq 0}$ and $R_B\colon\Xb\times\Ub\rightarrow\reals^n\times\reals^n_{\geq 0}$ which can produce characterization of an $\ell_\infty$ ball corresponding to the over-approximation of forward- and backward-in-time reachable sets, respectively, for every state-input pair picked from $\Xb\times\Ub$. 
\begin{algorithm}[t]
	\caption{Regression-based compression algorithm for finite abstractions}
	\label{alg:abstraction}
	\KwData{Forward dynamics $\bar\Sigma$ and learning rate $\lambda$}
	Compute backward dynamics $\bar\Sigma_B$ and the datasets $\ds_F$ and $\ds_B$ using Eqs.~\eqref{eq:BW_sys}, \eqref{eq:FW_ds}, and \eqref{eq:BW_ds}\\
	Train neural networks $\NN_F$ on the dataset $\ds_F$ and train $\NN_B$ on $\ds_B$ using the learning rate $\lambda$\\
	Compute the soundness errors $\eFb$ and $\eBb$ using Eq.~\eqref{eq:FW_soundness_err}\\
	Compute the final corrected representations $R_F$ and $R_B$ using Eqs.~\eqref{eq:R_F} and \eqref{eq:R_B}\\
	\KwResult{corrected neural representations $R_F$ and $R_B$}
\end{algorithm}
Our aim is  to use the expressive power of neural networks to represent the behavior of $\bar \Sigma$ such that the memory requirements significantly decrease. 

\begin{figure}[t]
	\centering
	\hspace{-.7cm}
	\includegraphics[scale=.7]{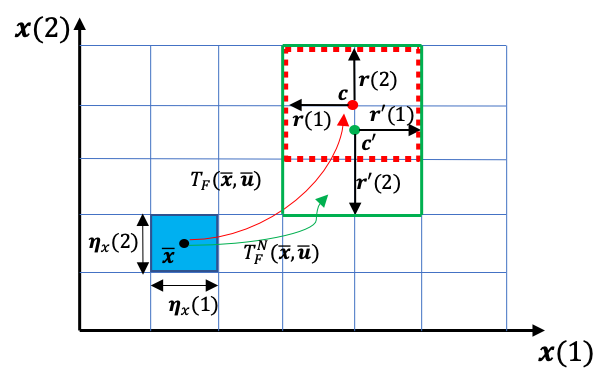}
	\caption{Comparing the set of successor states in the transition system $T_F$ and its representation $T_F^N$. We have
	$\boldsymbol c=c_F(\xbb,\ubb)$, 
	$r_F(\xbb,\ubb) = [\boldsymbol r(1),\boldsymbol r(2)]^\top$,
	$\boldsymbol c'=\NN_F^c(\xbb,\ubb)$
	and
	$\NN_F^r(\xbb,\ubb) = [\boldsymbol r'(1),\boldsymbol r'(2)]^\top$.}
	\label{fig:NN_vs_f}	
\end{figure}

Our compression scheme is summarized in Alg.~\ref{alg:abstraction}. We first compute the backward-in-time system $\bar\Sigma_B$ using Eq.~\eqref{eq:BW_sys}.
We then calculate the over-approximating $\ell_\infty$ ball for every state-input pair. Let $c_F(\xbb,\ubb)\in X$ and $r_F(\xbb,\ubb)\in\reals_{\geq 0}^n$ characterize the tightest $\ell_\infty$ ball such that 
 $$(\xbb,\ubb,\xbb')\in T_F\Leftrightarrow \|\xbb'-c_F(\xbb,\ubb)\|_{\infty}\leq r_F(\xbb,\ubb)-\etab_x/2.$$
This is illustrated in Fig.~\ref{fig:NN_vs_f} in two-dimensional space for a given state-input pair $(\xbb,\ubb)$. The dotted red rectangle corresponds to the hyper-rectangular reachable set. 
The center $c_F(\xbb,\ubb)$ and radius $r_F(\xbb,\ubb)$ are computed using the lower-left and upper-right corners of the reachable set denoted, respectively, by $g_{FL}(\xbb,\ubb)$ and $g_{FU}(\xbb,\ubb)$. Then, we have
$c_F(\xbb,\ubb)=(g_{FU}(\xbb,\ubb)+g_{FL}(\xbb,\ubb))/2$ and $r_F(\xbb,\ubb)=(g_{FU}(\xbb,\ubb)-g_{FL}(\xbb,\ubb))/2+\etab_x/2$.
 %
 At the end of the first step we have computed and stored the dataset
\begin{equation}\label{eq:FW_ds}
	\ds_F=\set{((\xbb, \ubb), (c_F(\xbb, \ubb), r_F(\xbb,\ubb)))\mid \xbb\in \Xb, \ubb\in \Ub}.
\end{equation}
Note that every data-point in $\ds_F$ consists of two pairs: one specifies a state-input pair $(\xbb,\ubb)$ and the other one characterizes the center and radius corresponding to the over-approximating $\ell_\infty$ disc $(c_F(\xbb, \ubb), r_F(\xbb,\ubb))$. Similarly, we need to store another dataset corresponding to the backward dynamics. First, we define $c_B(\xbb,\ubb)\in X$ and $r_B(\xbb,\ubb)\in\reals_{\geq 0}^n$ characterizing the tightest $\ell_\infty$ ball such that 
$$(\xbb,\ubb,\xbb')\in T_B\Leftrightarrow \|\xbb'-c_B(\xbb,\ubb)\|_{\infty}\leq r_B(\xbb,\ubb)-\eta_x/2.$$
The dataset corresponding to the backward dynamics is of the following form
\begin{equation}\label{eq:BW_ds}
	\ds_B=\set{((\xbb, \ubb), (c_B(\xbb, \ubb), r_B(\xbb,\ubb)))\mid \xbb\in \Xb, \ubb\in \Ub}.
\end{equation}
\begin{figure}[t]
		\includegraphics[scale=.25]{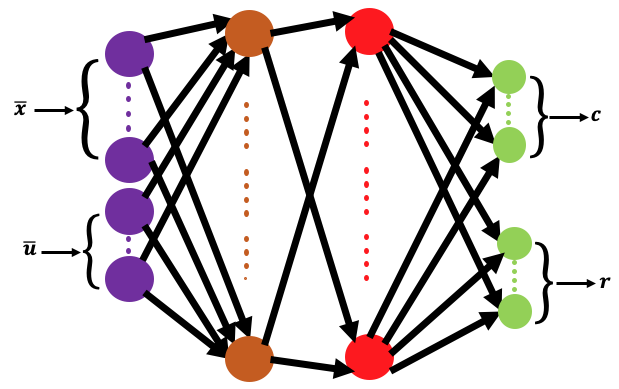}
	\caption{The regression-based configuration used in compression of abstractions. The input to the neural network includes state-input pair $(\xbb,\ubb)$, and the output includes the pair $(c,r)$ corresponding to the center and radius of the rectangular reachable set, respectively.
	Right: The classification-based representation of finite abstractions. The representation receives a state-input pair $(\xbb,\ubb)$. In the output, $\boldsymbol y_{lb}$ and $\boldsymbol y_{ub}$ correspond to the lower-left and upper-right corners for the rectangular reachable set.}
	\label{fig:NN_TS_regression}
\end{figure}
The size of $\ds_F$ and $\ds_B$ grows exponentially with the dimension of state space. Hence, we store both the datasets $\ds_F$ and $\ds_B$ (potentially) into the hard drive. Next, we take the datasets $\ds_F$ and $\ds_B$, for which we train neural networks $\NN_F$ and $\NN_B$, taking the state-input pairs $(\xbb,\ubb)$ as input and $(c_F(\xbb,\ubb),r_F(\xbb,\ubb))$ 
as output, and try to find an input-output mapping minimizing 
mean squared error (MSE). For systems with state and input spaces of dimensions $n$ and $m$, the input and output layers of both neural networks are of sizes $n+m$ and $2n$, respectively. The configuration of the neural networks which we used is illustrated in Fig.~\ref{fig:NN_TS_regression}. During training, we load batches of data from $\ds_F$ and $\ds_B$, which are stored on the the hard drive, into the RAM. We use the stochastic gradient descent (SGD) method to minimize MSE. 

As mentioned earlier, in contrast to the usual applications wherein neural networks are used to represent an unknown distribution, we have the full dataset and require computing representations which are \emph{sound} with respect to the input dataset.
%
A sound representation for the given finite abstractions produces reachable sets that include $T_F(\xbb,\ubb)$ for every state-input pair $(\xbb,\ubb)$. For instance, the solid green rectangle in Fig.~\ref{fig:NN_vs_f} contains the set of reachable states corresponding to $\NN_F(\xbb,\ubb)$ and contains the set of states included in the dotted red rectangle, i.e.,  $T_F(\xbb,\ubb)$. Therefore, we can say that the representation $\NN_F$ is sound for the pair $(\xbb,\ubb)$.
%
In order to guarantee \emph{soundness}, we need to compute the maximum error induced during the training process among all the training data points. To that end, we go over all the state-input pairs (which are stored on the hard drive) and compute the maximum error in approximating the centers of the $\ell_\infty$ balls, denoted by $\eFb^c, \eBb^c$ and radius $\eFb^r, \eBb^r$ corresponding to the forward and backward representations:
\begin{align*}\label{eq:FW_soundness_err_center}
	\eFb^c&=\max_{\xbb\in\Xb,\ubb\in\Ub}|c_F(\xbb,\ubb)-\NN_F^c(\xbb,\ubb)|,\quad \eFb^r=\max_{\xbb\in\Xb,\ubb\in\Ub}|r_F(\xbb,\ubb)-\NN_F^r(\xbb,\ubb)|.
\end{align*}
Similarly, for the backward dynamics,
\begin{align*}
	\eBb^c&=\max_{\xbb\in\Xb,\ubb\in\Ub}|c_B(\xbb,\ubb)-\NN_B^c(\xbb,\ubb)|,\quad \eBb^r=\max_{\xbb\in\Xb,\ubb\in\Ub}|r_B(\xbb,\ubb)-\NN_B^r(\xbb,\ubb)|.
\end{align*}
We define
\begin{equation}\label{eq:FW_soundness_err}
 \eFb=\eFb^c+\eFb^r,\quad \eBb=\eBb^c+\eBb^r,
\end{equation} 
and use the errors $\eFb$ and $\eBb$ to compute the \emph{corrected representations} $R_F$ and $R_B$, corresponding to $\NN_F$ and $\NN_B$, as described next. Let $R_F^c$ and $R_F^r$ correspond to the center and radius components of $R_F$. Similarly, $R_B^c$ and $R_B^r$ correspond to the center and radius components of $R_B$. For state-input pair $(\xbb,\ubb)\in\Xb\times \Ub$, we define
\begin{align}\label{eq:R_F}
	R_F^c(\xbb,\ubb)&=\NN_F^c(\xbb,\ubb),\quad R_F^r(\xbb,\ubb)=\NN_F^r(\xbb,\ubb)+\eFb,
\end{align}
for the forward dynamics, and
\begin{align}\label{eq:R_B}
	R_B^c(\xbb,\ubb)&=\NN_B^c(\xbb,\ubb),\quad R_B^r(\xbb,\ubb)=\NN_B^r(\xbb,\ubb)+\eBb,
\end{align}
for the backward dynamics.

Let us define the forward transition system computed using the trained neural network as follows
\begin{align}\label{eq:TS_NN_F}
	T_F^N = &\{(\xbb,\ubb, \xbb')\!\in\!\Xb\!\times\! \Ub\!\times\! \Xb\!\mid\!\xbb'\!\in\! \bar K(\NN_F^c(\xbb, \ubb)\!-\!\NN_F^r(\xb, \ubb)\!-\!\eFb,\NN_F^c(\xbb, \ubb)\!+\!\NN_F^r(\xbb, \ubb)\!+\!\eFb)\},
	\end{align}
	where $\NN_F^c(\cdot, \cdot)$, $\NN_F^r(\cdot, \cdot)$ denote the components of the output of $\NN_F(\cdot, \cdot)$ corresponding to the center and radius of disc, respectively. Similarly, we can define the transition system $T_B^N$ corresponding to the backward dynamics as follows
	\begin{align}\label{eq:TS_NN_B}
		T_B^N = &\{(\xbb,\ubb, \xbb')\!\in\!\Xb\!\times\! \Ub\!\times\! \Xb\!\mid\!\xbb'\!\in\! \bar K(\NN_B^c(\xbb, \ubb)\!-\!\NN_B^r(\xbb, \ubb)\!-\!\eBb,\NN_B^c(\xbb, \ubb)\!+\!\NN_B^r(\xbb, \ubb)\!+\!\eBb)\}.
	\end{align}
	 The following lemma states that we can use the trained neural networks to compute sound transition systems for both forward and backward dynamics. However, our synthesis approach does not require the computation of $T_F^N$ and $T_B^N$ and only uses the compressed representations $\NN_F$ and $\NN_B$.
	\begin{lemma}\label{lem:TS_soundness}
		Transition systems $T_F^N$ and $T_B^N$ computed by \eqref{eq:TS_NN_F} and \eqref{eq:TS_NN_B} are sound for $T_F$ and $T_B$, i.e., we have $T_F\subseteq T_F^N$ and $T_B\subseteq T_B^N$.
	\end{lemma}
To reduce the level of conservativeness, we require that $T_F^N$ and $T_B^N$ do not contain too many additional edges compared to $T_F$ and $T_B$. The \emph{mismatch rate} of the forward and backward dynamics are defined as 
\begin{equation*}\label{eq:mismatch_rate_FW_BW}
	d_F\coloneq \frac{|T_F^N\setminus T_F|}{|T_F|}, \quad d_B\coloneq \frac{|T_B^N\setminus T_B|}{|T_B|}.
\end{equation*}
If the trained representations are accurate, the mismatch rate is low, which results in a \emph{less restrictive} representation. 
\begin{remark}\label{rem:cls_method}
The method proposed in this section formulates the computation of the representations as a regression problem, wherein the representative neural networks are supposed to predict the center and radius corresponding to $\ell_\infty$ reachable sets.
In Sec.~\ref{sec:TS_compression_classification}, we describe a classification-based formulation for compressing finite abstractions, wherein the representative neural networks are supposed to predict the vectorized indices corresponding to the lower-left and upper-right corners of the reachable set. We experimentally show that this second formulation, while being more memory demanding, provides a less conservative representation compared to the formulation discussed in this section.
\end{remark}

%% file: cls_TS_compression.tex
\begin{figure}[t]
		\includegraphics[scale=.25]{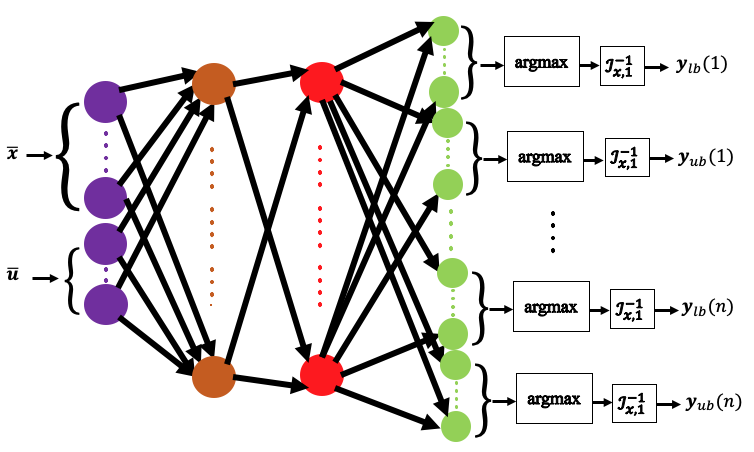}
	\caption{The classification-based representation of finite abstractions. The representation receives a state-input pair $(\xbb,\ubb)$. In the output, $\boldsymbol y_{lb}$ and $\boldsymbol y_{ub}$ correspond to the lower-left and upper-right corners for the rectangular reachable set.}
	\label{fig:NN_TS_cls}
\end{figure}
\subsection{Classification-Based Computation of Representations for Finite Abstractions}\label{sec:TS_compression_classification}
We proposed in Sec.~\ref{subsec:TS_compression_regression} a formulation for training neural networks that can \emph{guess} at any given state-input pair the center and radius of a hyper-rectangular over-approximation of the reachable states. This guess is then \emph{corrected} using the computed soundness errors. A nice aspect of this formulation is that we only need to store the trained representations and their corresponding soundness errors. However, the result of using the soundness errors to correct the output values produced by the neural networks may give a very conservative over-approximation of the reachable sets, even when the trained representations have a very good performance on a large subset of the state-input pairs, since the soundness errors must be computed over \emph{all} state-input pairs.

\begin{algorithm}[t]
	\caption{Computing classification-based representations of finite abstractions}
	\label{alg:compression_cls}
	\KwData{Forward dynamics $\bar\Sigma$ and learning rate $\lambda$}
	Compute backward dynamics $\bar\Sigma_B$ and the datasets $\ds_F$ and $\ds_B$ using Eqs.~\eqref{eq:BW_sys}, \eqref{eq:FW_ds_cls}, and \eqref{eq:BW_ds_cls}\\
	Train neural networks $\NN_F$ and $\NN_B$ on the datasets $\ds_F$ and $\ds_B$ using the learning rate $\lambda$\\
	Compute the set of misclassified state-input pairs $E_F$ and $E_B$ as in Eq.~\eqref{eq:miss}\\
	Compute the set of transitions $\tilde{\NN}_F$ and $\tilde{\NN}_B$ associated with $E_F$ and $E_B$ as in  Eq.~\eqref{eq:TS_to_be_stored}\\
	Compute the corrected neural representations $R_F$, $R_B$ using Eqs.~\eqref{eq:R_F_def_cls}, \eqref{eq:R_B_def_cls}\\
	\KwResult{$R_F$, $R_B$}
\end{algorithm}
In this section, we provide an alternative formulation for computing a compressed representation of a given abstraction. Intuitively, our idea is to train neural network representations which can guess for any given state-input pair the \emph{vectorized indices} corresponding to the lower-left and upper-right corner points of the hyper-rectangular reachable set. The architecture of the representation is shown in Fig.~\ref{fig:NN_TS_cls}. As illustrated, for every state-input pair $(\xbb,\ubb)\in \Xb\times \Ub$, the output of the representation gives 
the lower-left and upper-right corners of the rectangular set that is reachable by taking the control input $\ubb$ at the state $\xbb$. 
Alg.~\ref{alg:compression_cls} describes our classifier-based compression scheme for finite abstractions. 
  We first compute the backward system $\bar\Sigma_B$ using Eq.~\eqref{eq:BW_sys}. We then compute the training datasets for both the forward and backward systems $\bar\Sigma$ and $\bar\Sigma_B$. For 
  $\bar\Sigma$, let $g_{FU}\colon\Xb\times\Ub\rightarrow \Xb$ and $g_{FL}\colon\Xb\times\Ub\rightarrow \Xb$ denote the mappings from the state-input pair $(\xbb,\ubb)\in\Xb\times\Ub$ into the corresponding upper-right and lower-left corners of the rectangular reachable set from $(\xbb,\ubb)$. We define
  $z_F\colon \Xb\times \Ub\rightarrow\{0,1\}^{2\sum_{i=1}^n|\Xb(i)|}$ with $|\Xb(i)|$ being the cardinality of the projection of $\Xb$ along the $i^{\text{th}}$ axis and $z_F(\xbb,\ubb)(l)=1$ if and only if
\begin{align*}
	& l=2\sum_{k=1}^{i-1}|\Xb(k)|+\Ixi(g_{FL}(\xbb,\ubb)(i)) \text{ or } l=2\sum_{k=1}^{i-1}|\Xb(k)|+|\Xb(i)|+\Ixi(g_{FU}(\xbb,\ubb)(i)),
\end{align*}
for some $i\in\{1,2,\ldots,n\}$. The indexing function $\Ixi\colon \Xb(i)\rightarrow [1;|\Xb(i)|]$ maps every element of $\Xb(i)$ into a \emph{unique} integer index in the interval $[1;|\Xb(i)|]$.
The training dataset for $\bar\Sigma$ is defined as
\begin{equation}
\label{eq:FW_ds_cls}
\ds_F\coloneq
	\{(\xbb,\ubb, z_F(\xbb,\ubb)) \,|\, \xbb\in\Xb \text{ and }\ubb\in \Ub\}.
\end{equation}
Intuitively, each element of the dataset $\ds_F$ contains a state-input pair $(\xbb,\ubb)$ and a vector $\boldsymbol h\in \{0,1\}^{2\sum_{i=1}^n|\Xb(i)|}$ that has $1$ only at the entries corresponding to $\Ixi(g_{FL}(\xbb,\ubb)(i))$ and $\Ixi(g_{FU}(\xbb,\ubb)(i))$ for $i\in\{1,2,\ldots,n\}$.
Similarly, we define
$z_B\colon \Xb\times \Ub\rightarrow\{0,1\}^{2\sum_{i=1}^n|\Xb(i)|}$ for $\bar\Sigma_B$ such that $z_B(\xbb,\ubb)(l)=1$ if and only if
\begin{align*}
	& l=2\sum_{k=1}^{i-1}|\Xb(k)|+\Ixi(g_{BL}(\xbb,\ubb)(i))\text{ or } l=2\sum_{k=1}^{i-1}|\Xb(k)|+|\Xb(i)|+\Ixi(g_{BU}(\xbb,\ubb)(i)),
\end{align*}
for some $i\in\{1,2,\ldots,n\}$.
The training dataset for the backward dynamics is also defined similarly as
\begin{equation}
\label{eq:BW_ds_cls}
	\ds_B\coloneq
	\{(\xbb,\ubb, z_B(\xbb,\ubb)) \,|\, \xbb\in\Xb \text{ and }\ubb\in \Ub\}.
\end{equation}

Once the training datasets are ready, we train the neural networks $\NN_F$ and $\NN_B$  respectively on the datasets $\ds_F$ and $\ds_B$. Note that the output layer of $\NN_F$ and $\NN_B$ will be a vector of size $2\sum_{i=1}^n|\Xb(i)|$, while the final output of the representations are of size $2n$ (cf. Fig.~\ref{fig:NN_TS_cls}). These final outputs give an approximation of the coordinates of the lower-left and upper-right corners of the reachable set corresponding to the pair $(\xbb,\ubb)$.
Note that, because $\Xb$ was computed by equally partitioning over $X$, both the indexing function $\Ixi$ and its inverse can be implemented in a memory-efficient way using \textsf{floor} and \textsf{ceil} operators. 
We then evaluate the performance of the trained neural networks $\NN_F$ and $\NN_B$.
Let $\rho_{FL}(\xbb,\ubb)$ and $\rho_{FU}(\xbb,\ubb)$ denote respectively the estimated lower-left and upper-right corners of the reachable set estimated by $\NN_F$. Define $\rho_{BL}(\xbb,\ubb)$ and $\rho_{BU}(\xbb,\ubb)$ similarly for $\NN_B$, and let the set of misclassified state-input pairs be
\begin{align}
	E_F&\coloneq \{(\xbb,\ubb)\in\Xb\times\Ub\mid\ T_F(\xbb,\ubb)\setminus \llbracket \rho_{FL}(\xbb,\ubb), \rho_{FU}(\xbb,\ubb)\rrbracket_{\eta_x}\neq \emptyset\}\nonumber\\
	E_B&\coloneq\{(\xbb,\ubb)\in\Xb\times\Ub\mid\ T_B(\xbb,\ubb)\setminus \llbracket \rho_{BL}(\xbb,\ubb), \rho_{BU}(\xbb,\ubb)\rrbracket_{\eta_x}\neq \emptyset\}.\label{eq:miss}
\end{align}
 The \emph{soundness error} of $\NN_F$ and $\NN_B$ can be considered as their misclassification rate:
\begin{equation}\label{eq:cls_soundness_errors}
	err_F\coloneq\frac{|E_F|}{|\Xb\times\Ub|} \quad\text{ and }\quad err_B\coloneq\frac{|E_B|}{|\Xb\times\Ub|}.
\end{equation}
For the misclassified 
pairs in $E_F$ and $E_B$, we extract the related transitions in the abstraction:
\begin{align}\label{eq:TS_to_be_stored}
	\tilde{\NN}_F & \!\coloneq\!\set{(\xbb,\ubb,\xbb')\!\mid\! (\xbb,\ubb)\!\in\! E_F,\xbb'\!\in\! T_F(\xbb,\ubb)},
	\tilde{\NN}_B  \!\coloneq\!\set{(\xbb,\ubb,\xbb')\!\mid\! (\xbb,\ubb)\in E_B,\xbb'\!\in\! T_B(\xbb,\ubb)}\!.\!
\end{align}
Finally, we \emph{correct} the output of neural network representations to maintain soundness
\begin{align}
	\label{eq:R_F_def_cls}
	 R_F(\xbb,\ubb)\!\coloneq\!
	\begin{cases}
		\llbracket\rho_{FL}(\xbb,\ubb),\rho_{FU}(\xbb,\ubb)\rrbracket_{\eta_x}\!\!\!\! &\text{if}\; (\xbb,\ubb)\!\notin\! E_F\\
		\tilde{\NN}_F(\xbb,\ubb)  &\text{if}\;(\xbb,\ubb)\!\in\! E_F,
	\end{cases}\\
	\label{eq:R_B_def_cls}
	 R_B(\xbb,\ubb)\!\coloneq\!
	\begin{cases}
		\llbracket\rho_{BL}(\xbb,\ubb),\rho_{BU}(\xbb,\ubb)\rrbracket_{\eta_x} \!\!\!\! &\text{if}\; (\xbb,\ubb)\!\notin\! E_B\\
		\tilde{\NN}_B(\xbb,\ubb) &\text{if}\;(\xbb,\ubb)\!\in\! E_B.
	\end{cases}
\end{align}
Note that these corrected neural representations are memory efficient only if the misclassification rates are small, i.e., the size of $E_F$ and $E_B$ are small compared with $\Xb\times\Ub$.

%% file: synthesis.tex

\subsection{On-the-Fly Synthesis}
\label{subsec:on-the-fly_synthesis}
\begin{algorithm}[t]
	\caption{Controller synthesis algorithm}
	\label{alg:synthesis}
	\KwData{Set $W_0\subseteq \Xb$ and the corrected neural representations $R_F$ and $R_B$}
	Initialize $C\leftarrow \emptyset$,  $P_0\leftarrow W_0$, $\Gamma_0\leftarrow \emptyset$ and $i\leftarrow 0$\\
	\While{$W_i\neq \emptyset$}{
		Compute the candidate pool $S_{i}$ using Eq.~\eqref{eq:candid_pool}\\
		Compute the set of new winning states $W_{i+1}$ using Eq.~\eqref{eq:NN_post} and add them to the winning set ($P_{i+1}\leftarrow P_i\cup W_{i+1}$)\\
		Compute the set of new state-input pairs $\Gamma_{i+1}$ using Eq.~\eqref{eq:controller_added_pairs} and add them to the controller ($C\leftarrow C\cup \Gamma_{i+1}$)\\
		$i\leftarrow i+1$}
	$L\leftarrow P_{i}$\\
	\KwResult{Controller $C$ and its winning set $L$}
\end{algorithm}
In the previous subsection, we described the computation of the compressed representations corresponding to the forward and backward dynamics for finite abstractions. 
In this subsection, we use these representations in order to synthesize formally correct controllers. 

Our synthesis procedure is provided in Alg.~\ref{alg:synthesis}. It takes the representations $R_F$ and $R_B$ to synthesize a controller which fulfills the given reach-avoid specification. 
Let
$$W_0=\set{\xbb\in\Xb\mid \llbracket\xbb-\boldsymbol{\eta_x}/2, \xbb+\boldsymbol{\eta_x}/2\rrbracket\subseteq \goal}$$
be a discrete under-approximation of the target set $\goal$. We take $W_0$ as the input and perform a fixed-point computation to solve the given reach-avoid game. We initialize the winning set and controller with $P_0=W_0$ and $C=\emptyset$, and in each iteration, we add the new winning set of states and state-input pairs, respectively, into the overall winning set and the controller, until no new state is found ($W_{i+1}=\emptyset$). 

Let $W_i$ be the set of \emph{new} winning states in the beginning of the $i^{th}$ iteration. Further, we denote the set of winning states in the beginning of the $i^{th}$ iteration by $P_i=\bigcup_{k=0}^i W_k$. In every iteration, for every $\xbb\in W_i$ and $\ubb\in \Ub$, we compute the backward over-approximating $\ell_\infty$ ball and discretize it to get the \emph{candidate pool} $S_{i}$ defined as 
\begin{equation}\label{eq:candid_pool}
	S_i\coloneq \bigcup_{\ubb\in \Ub} Y_i(\ubb),
\end{equation}
with
\begin{align*}
	&Y_{i}(\ubb)\coloneq\bigcup_{\xbb\in W_i}(\Xb\cap
	\bar K(R_B^c(\xbb, \ubb)-R_B^r(\xbb, \ubb),R_B^c(\xbb, \ubb)+R_B^r(\xbb, \ubb))),
\end{align*}
where $R_B^c(\cdot, \cdot)$, $R_B^r(\cdot, \cdot)$ denote the components of the output of $R_B(\cdot, \cdot)$ corresponding to the center and radius of the $\ell_\infty$ ball, respectively.
Note that we compute the candidate pool by running $R_B$ over $W_i$ \emph{instead of} $P_i$. This is computationally beneficial, because $|W_i|\leq |P_i|$. Next lemma shows that $S_i$ includes the \emph{whole} set of new winning states $W_{i+1}$.
\begin{lemma}\label{lem:candidate_pool}
	Let the set of candidates $S_{i}$ be as defined in Eq.~\eqref{eq:candid_pool}. Then, we have $W_{i+1}\subseteq S_{i}$ for all $i\geq 0$.
\end{lemma}
\begin{proof}
	We prove this lemma by contradiction. Suppose that $W_{i+1}\not\subseteq S_i$. Then there exists at least one $\xbb^\ast\in W_{i+1}\setminus S_i$. Since $\xbb^\ast\in W_{i+1}$, we know that there exists at least one $\ubb^\ast\in\Ub$ such that $T_f(\xbb^\ast,\ubb^\ast)\subseteq P_i$ and $\xbb^\ast\notin P_i$. Moreover, since $\xbb^\ast\notin S_i$, by Eq.~\eqref{eq:candid_pool} we get $T_f(\xbb^\ast,\ubb^\ast)\cap W_i=\emptyset$. 
	So, $T_f(\xbb^\ast,\ubb^\ast)\subseteq P_i\setminus W_i=P_{i-1}$. This gives $\xbb^\ast\in P_i$, which is a contradiction. This completes the proof.
\end{proof}

 Now, we can use $R_F$, which represents the forward transition system, in order to choose the \emph{legitimate} candidates out of $S_{i}$ and add the new ones to $W_{i+1}$.
 Let
 $$A=\set{\xbb\in\Xb\mid \llbracket\xbb-\boldsymbol{\eta_x}/2, \xbb+\boldsymbol{\eta_x}/2\rrbracket\cap \avoid\neq \emptyset}$$
 be a discrete over-approximation over the set of obstacles.
 The next lemma states that we can use the representation $R_F$ to compute $W_{i+1}$. 
\begin{lemma}\label{lem:legitimate_selection}
		The set of states added to the winning set in the $i^{th}$ step can be computed as \begin{align}\label{eq:NN_post}
		&W_{i+1} =\{\xbb\in S_{i}\mid \exists \ubb\in\Ub\; s.t.\bar K(R_F^c(\xbb, \ubb)-R_F^r(\xbb, \ubb),R_F^c(\xbb, \ubb)+R_F^r(\xbb, \ubb))\subseteq P_i	
		\}\setminus (P_i\cup A).
	\end{align} 
\end{lemma}
\begin{proof}
	To prove this lemma, we denote $G=\{\xbb\in S_{i}\mid \exists \ubb\in\Ub\; s.t.\;\bar K(R_F^c(\xbb, \ubb)-R_F^r(\xbb, \ubb),R_F^c(\xbb, \ubb)+R_F^r(\xbb, \ubb))\subseteq P_i	
	\}\setminus (P_i\cup A)$, and show $W_{i+1}\subseteq G$ and $G\subseteq W_{i+1}$. The second direction ($G\subseteq W_{i+1}$) holds by definition. To prove the first direction ($W_{i+1}\subseteq G$), we note that $G\subseteq  S_{i}$ and further, by the result of Lemma.~\ref{lem:candidate_pool}, we have $W_{i+1}\subseteq S_{i}$. Assume $W_{i+1}\not\subseteq G$. Then there should exist at least one $\xbb^\ast\in W_{i+1}\setminus G$. Note that $\xbb^\ast\in S_{i}\setminus G$. Since $\xbb^\ast\in S_i$, we get that there exists at least one $\ubb^\ast\in\Ub$ for which $T_F(\xbb^\ast,\ubb^\ast)\subseteq W_i$. Also, because $\xbb^\ast\notin G$, we have $T_F(\xbb^\ast,\ubb^\ast)\not\subseteq W_i$, which is a contradiction. Therefore, $W_{i+1}\subseteq G$. Hence the proof ends.
\end{proof}
In each iteration, we calculate $\Gamma_i$, which is the set of new state-input pairs that must be added into the controller, and is defined as
\begin{align}\label{eq:controller_added_pairs}
	\Gamma_{i+1}=&\{(\xbb,\ubb)\mid
		\xbb\in W_{i+1},\;\bar{K}(R_F^c(\xbb, \ubb)-R_F^r(\xbb, \ubb),R_F^c(\xbb, \ubb)+R_F^r(\xbb, \ubb))\subseteq P_i	\}.
\end{align}
Finally, If $W_{i+1}=\emptyset$, we can terminate the computations as we already have computed the winning set and the controller. Otherwise, we add $W_i$ and $\Gamma_i$ into the overall winning set ($P_{i+1}\leftarrow P_i\cup W_{i+1} $) and controller ($C\leftarrow C\cup \Gamma_{i+1} $)  and restart the depicted process.

%% file: deployment.tex

\section{Deployment}\label{sec:deployment}

Once the controller $C$ is computed such that $C\parallel \bar\Sigma$ realizes the given specification $\Phi$, we need to deploy $C$ onto an embedded controller platform, e.g., a microcontroller. Since such embedded controller platforms generally have 
a small on-board memory, we would like to minimize the size of the stored controller.

We define the \emph{set of valid control inputs} corresponding to $\xbb$ as $C(\xbb)=\set{\ubb\mid(\xbb,\ubb)\in C}$. The approach we proposed for finding representations for the finite abstractions may not work, since 
we are not allowed to over-approximate $C(\xbb)$, and thus the set of valid control inputs is not representable as a compact $\ell_\infty$ ball described by its center and radius. 
The following example illustrates a disconnected $C(\xbb)$, which cannot be represented by an $\ell_\infty$ ball.
\begin{figure}
	\centering
		\includegraphics[scale=1]{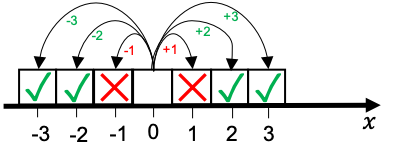}
	\caption{Illustration of a disconnected set of valid control inputs.}
	\label{fig:ex_disconnected}
\end{figure}
\begin{figure}
	\centering
	\includegraphics[scale=.23]{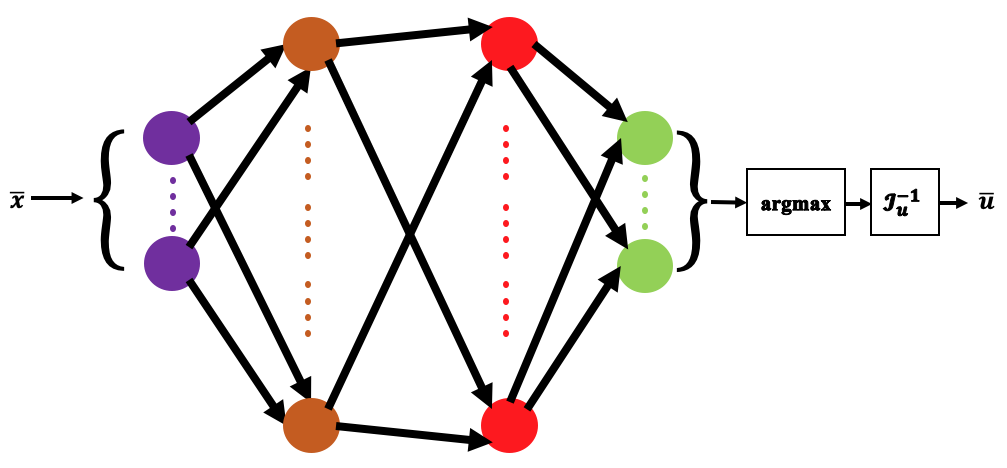}
	\caption{The configuration used in compression of controllers. Given a state $\xbb$, the representation produces a corresponding control input $\ubb$.}
	\label{fig:deployment}
\end{figure}
\begin{example}
Consider a system with one-dimensional state and input spaces ($n=m=1$). Fig.~\ref{fig:ex_disconnected} illustrates the set of transitions starting from the white middle box ($\xbb=0$). Let the boxes with green check mark and red cross mark correspond to the target and obstacle states and $C$ be the controller for the corresponding reach-avoid specification. Then, we have $\set{(0,2),(0,3),(0,-2),(0,-3)}\subseteq C$ and $C(0)=\set{-2,-3,2,3}$. It is clear that $C(0)$ is a disconnected set, which is not characterizable by an $\ell_\infty$ ball.
\end{example}
In contrast to the symbolic regression method proposed in \cite{ZAPREEV:2018}, we formulate the controller compression problem as a classification task, that is, we train a neural network which assigns every state to a list of scores over the set of control inputs, and picks the control input with the highest score. The configuration of the neural network is illustrated in Fig.~\ref{fig:deployment}. The justification for our formulation is that any representation for the controller can only perform well if it is trained over a dataset which respects the \emph{continuity property}, i.e., neighboring states are not mapped into control input values which are very different from each other. A representation that respects the continuity property corresponds to a low continuity index (see Eq.~\eqref{eq:cont_idx}). During the training phase, we keep all the valid control inputs and let the training process to choose which value respects the continuity property more, by minimization of the cost function. Therefore, our formulation automatically takes care of the \emph{redundancy} problem by mapping a neighborhood in the state space into close-in-value control inputs to respect the continuity requirement of the trained representation. The reason that our formulation does not correspond to a \emph{standard} classification setting is that during the training phase a \emph{non-uniform} number of  labels (corresponding to the control input values in the output stage of the neural network) per input (corresponding to the state values at the input layer of the neural network) are considered as valid, while we only will consider \emph{one} label---corresponding to the highest score---as the trained representation's choice during the runtime.
\begin{remark}
	In order to formulate the problem of finding a neural-network-based representation for the controller as a regression problem, first the training data must be \emph{pre-processed} such that the continuity property is respected, i.e., the set of valid control-inputs per each state is pruned so that neighboring states are mapped to close-in-value control inputs. However, this pre-processing is time consuming and does not work efficiently in practice (see, e.g., \cite{ZAPREEV:2018, GIRARD:2012}).
\end{remark}

\begin{algorithm}[t]
	\caption{Compression algorithm for the controller}
	\label{alg:controller-compression}
	\KwData{Controller $C$, learning rate $\lambda$}
	Compute the dataset $\ds_C$ using Eq.~\eqref{eq:cont_ds}\\
	Train the neural network $\NN_C$ on the dataset $\ds_C$ using the learning rate $\lambda$\\
	Compute the set of state-input pairs  $\tilde C$ using Eq.~\eqref{eq:cont2nd_to_be_stored}\\
	Compute $\hat C$ using Eq.~\eqref{eq:C_hat_def}\\
	\KwResult{Corrected neural representation $\hat C$}
\end{algorithm}

Alg.~\ref{alg:controller-compression} summarizes the proposed procedure for computing a compressed representation for the original controller. 
In the first step, we need to store the training set
\begin{align}\label{eq:cont_ds}
	\ds_C =\{(\xbb,\boldsymbol h(\xbb))\mid &(\xbb,\ubb)\in C \Leftrightarrow \boldsymbol h(\xbb)(\Iu(\ubb))=1,\;(\xbb,\ubb)\notin C \Leftrightarrow \boldsymbol h(\xbb)(\Iu(\ubb))=0\},
\end{align}
where $\Iu\colon U\rightarrow [1;|\Ub|]$ is an \emph{indexing function} for the control set $\Ub$, which assigns every value in $\Ub$ into a \emph{unique} integer in the interval $[1;|\Ub|]$. Intuitively, each point in the dataset $\ds_C$ contains a state $\xbb\in L$ and a vector $\boldsymbol h(\xbb)$ which is of length $|\Ub|$ and has ones at the entries corresponding to the valid control inputs and zeros elsewhere.

Once the training dataset is ready, we can train a neural network $\NN_C$ which takes $\xbb\in\Xb$ as input and approximates $\Iu^{-1}(argmax(\boldsymbol h(\xbb)))$ in the output, where $\Iu^{-1}(\cdot)$ denotes the inverse of the indexing function used in Eq.~\eqref{eq:cont_ds}.
\begin{remark}
	Note that the output layer of $\NN_C$ has to be of size $|\Ub|$ and for every $\xbb\in L$, 
	we consider the value  $\Iu^{-1}(argmax(\NN_C(\xb)))$ as the final control input assigned by $\NN_C$ to the state $\xbb$. Moreover, because $\Ub$ was computed by equally partitioning over $U$, both the indexing function $\Iu$ and its inverse can be implemented in a memory-efficient way using \emph{\textsf{floor}} and \emph{\textsf{ceil}} functions. 
\end{remark}

Once the neural network $\NN_C$ is trained, we evaluate its performance 
by finding all the states $\xbb$ at which using $\NN_C$ produces an \emph{invalid} control input, i.e., 
$$E=\set{\xbb\in L\mid \Iu^{-1}(argmax(\NN_C(\xbb)))\notin C(\xbb)}.$$ 
The misclassification rate of the trained classifier $\NN_C$ is defined as:
\begin{equation}\label{eq:err_C}
	err_C=\frac{|E|}{|L|}.
\end{equation}
In order to maintain the guarantee provided by the original controller $C$, it is very important to \emph{correct} the output of the trained representation, so that it outputs a valid control input at \emph{every} state. In case the misclassification rate is small, we can store $\NN_C$ together with $\tilde{C}$, where 
\begin{equation}\label{eq:cont2nd_to_be_stored}
 \tilde C =\set{(\xbb,\ubb)\mid \xbb\in E,\;\ubb\in C(\xbb)}.
\end{equation}
The final deployable controller $\hat C$ consists of both $\NN_C$ and $\tilde C$, and is defined as 
\begin{equation}
	\label{eq:C_hat_def}
	\hat C(\xbb)\coloneq \left\{
	\begin{array}{lr}
		\Iu^{-1}(argmax(\NN_C(\xbb)))\quad \text{if}\; \xbb\notin E\\
		\tilde{C}(\xbb)\qquad\qquad\qquad\qquad\; \text{if}\;\xbb\in E.
	\end{array}\right.
\end{equation}
\begin{lemma}\label{lemma:conreoller
	}
	Let $\hat C$ be as defined in Eq.~\eqref{eq:C_hat_def}. The winning domain of both $\hat C\parallel \bar\Sigma$ and $C\parallel \bar\Sigma$ for satisfying a specification $\Phi$ is the same.
\end{lemma}
\begin{remark}
Our deployment method preserves soundness.
The input to our deployment approach is a formally guaranteed controller computed by any abstraction-based method. We train a neural representation that maps the states to a control input. This control input is valid for majority of the states. For the states that the control input is not valid, we keep the set of valid control inputs from the original controller and store them as a small look-up table. Therefore, the final corrected neural controller in Eq.~\eqref{eq:C_hat_def} is sound with respect to the original controller.	
\end{remark}

%% file: experiments.tex

\section{Experimental Evaluation}\label{sec:experiments}
\begin{table*}
	\tiny
	\centering
	\caption{Catalog of models used to generate the finite abstractions in Sec.~\ref{sec:experiments}.
		\textmd{}
		\label{tab:case_studies}
	}
	
	\renewcommand{\arraystretch}{1.2}
	\setlength{\tabcolsep}{0.2em} 
	\begin{tabular}{l|c|c|c|c|c|c|c|c|c}
		\toprule
		\multirow{2}{*}{Case study}&\multirow{2}{*}{Dynamics of the model}& \multicolumn{4}{c}{Configuration (1)}&\multicolumn{4}{|c}{Configuration (2)}\\\cmidrule(l){3-10}
		&&$X$&$U$&$\eta_x$&$\eta_u$&$X$&$U$&$\eta_x$&$\eta_u$\\
		\midrule
		
		\multirow{1}{*}{2D car}& \multirow{2}{*}{$\begin{bmatrix}\dot\xb(1)\\\dot\xb(2)\end{bmatrix}\in\begin{bmatrix}\ub(1)\\\ub(2)\end{bmatrix}+W$}&$[0,5]^2$&$[-1,1]^2$&\multirow{2}{*}{$\begin{bmatrix}0.05\\0.05\end{bmatrix}$}&\multirow{2}{*}{$\begin{bmatrix}0.23\\0.23\end{bmatrix}$}&$[0,10]^2$&$[-2,2]^2$&\multirow{2}{*}{$\begin{bmatrix}0.025\\0.025\end{bmatrix}$}&\multirow{2}{*}{$\begin{bmatrix}0.3\\0.3\end{bmatrix}$}\\{$(\xb(1),\xb(2))$- position}&&&&&&&&&\\{$(\ub(1),\ub(2))$- speed}&$\tau=0.4$, $W=[-0.025,0.025]^2$&&&&&&&&\\
		
		\midrule
		
		\multirow{1}{*}{3D car}& \multirow{3}{*}{$\begin{bmatrix}\dot\xb(1)\\\dot\xb(2)\\\dot \xb(3)\end{bmatrix}\in\begin{bmatrix}\ub(1)\cos(\xb(3))\\\ub(1)\sin(\xb(3))\\\ub(2)\end{bmatrix}+W$}&$[0,5]^2\times$&$[-1,1]^2$&\multirow{3}{*}{$\begin{bmatrix}0.2\\0.2\\0.2\end{bmatrix}$}&\multirow{2}{*}{$\begin{bmatrix}0.3\\0.3\end{bmatrix}$}&$[0,10]^2\times$&$[-1.5,1.5]\times$&\multirow{3}{*}{$\begin{bmatrix}0.1\\0.1\\0.1\end{bmatrix}$}&\multirow{2}{*}{$\begin{bmatrix}0.23\\0.23\end{bmatrix}$}\\{$(\xb(1),\xb(2))$- position}&&$[-1.6, 1.6]$&&&&$[-\pi, \pi]$&$[-1, 1]$&&\\
		$\xb(3)$- angle&&&&&&&&&\\
		$\ub(1)$- speed&&&&&&&&&\\
		$\ub(2)$- turn rate&$\tau=0.3$, $W=\set{\textbf 0}$&&&&&&&&\\
		
		\midrule
		
		\multirow{1}{*}{4D car}& \multirow{4}{*}{$\begin{bmatrix}\dot\xb(1)\\\dot\xb(2)\\\dot \xb(3)\\\dot \xb(4)\end{bmatrix}\in\begin{bmatrix}\xb(4)\cos(\xb(3))\\\xb(4)\sin(\xb(3))\\\ub(1)\\\ub(2)\end{bmatrix}+W$}&{$[0,5]^2\times$}&$[-1,1]^2$&\multirow{4}{*}{$\begin{bmatrix}0.2\\0.2\\0.2\\0.2\end{bmatrix}$}&\multirow{2}{*}{$\begin{bmatrix}0.3\\0.3\end{bmatrix}$}&{$[0,10]^2\times$}&$[-2,2]^2$&\multirow{4}{*}{$\begin{bmatrix}0.2\\0.2\\0.2\\0.2\end{bmatrix}$}&\multirow{2}{*}{$\begin{bmatrix}0.2\\0.2\end{bmatrix}$}\\{$(\xb(1),\xb(2))$- position}&&$[-1.6,1.6]\times$&&&&$[-\pi,\pi]\times$&&&\\
		$\xb(3)$- angle&&$[-1,1]$&&&&$[-1,1]$&&&\\
		$\xb(4)$- speed&&&&&&&&&\\
		$\ub(1)$- turn rate&&&&&&&&&\\
		$\ub(2)$- acceleration control&$\tau=0.5$, $W=\set{\textbf 0}$&&&&&&&&\\
		
		\midrule
		
		\multirow{1}{*}{5D car}& \multirow{4}{*}{$\begin{bmatrix}\dot\xb(1)\\\dot\xb(2)\\\dot \xb(3)\\\dot \xb(4)\\\dot\xb(5)\end{bmatrix}\in\begin{bmatrix}\xb(4)\cos(\xb(3))\\\xb(4)\sin(\xb(3))\\\xb(5)\\\ub(1)\\\ub(2)\end{bmatrix}+W$}&{$[0,5]^2\times$}&$[-1,1]^2$&\multirow{4}{*}{$\begin{bmatrix}0.2\\0.2\\0.2\\0.2\\0.2\end{bmatrix}$}&\multirow{2}{*}{$\begin{bmatrix}0.3\\0.3\end{bmatrix}$}&{$[0,10]^2\times$}&$[-2,2]^2$&\multirow{4}{*}{$\begin{bmatrix}0.1\\0.1\\0.1\\0.1\\0.1\end{bmatrix}$}&\multirow{2}{*}{$\begin{bmatrix}0.2\\0.2\end{bmatrix}$}\\{$(\xb(1),\xb(2))$- position}&&$[-1.6,1.6]\times$&&&&$[-\pi,\pi]\times$&&&\\
		$\xb(3)$- angle&&$[-1,1]\times $&&&&$[-1,1]\times$&&&\\
		$\xb(4)$- speed&&$[-1,1]$&&&&$[-1,1]$&&&\\
		$\xb(5)$- turn rate&&&&&&&&&\\
		$\ub(1)$- acceleration &&&&&&&&&\\
		$\ub(2)$- angular acceleration &$\tau=0.5$, $W=\set{\textbf 0}$&&&&&&&&\\
		
		\midrule
		\multirow{1}{*}{Inverted pendulum}& \multirow{2}{*}{$\begin{bmatrix}\dot\xb(1)\\\dot\xb(2)\end{bmatrix}\in\begin{bmatrix}\xb(2)\\\frac{g}{L}\sin(\xb(1))+ \frac{1}{mL^2}\ub(1)\end{bmatrix}+W$}&$[-\pi,\pi]\times$&$[-1,1]$&\multirow{2}{*}{$\begin{bmatrix}0.2\\0.2\end{bmatrix}$}&\multirow{1}{*}{$\begin{bmatrix}0.3\end{bmatrix}$}&$[-\pi,\pi]\times$&$[-1,1]$&\multirow{2}{*}{$\begin{bmatrix}0.1\\0.1\end{bmatrix}$}&\multirow{1}{*}{$\begin{bmatrix}0.3\end{bmatrix}$}\\
		$\xb(1)$- angle&&$[-2,2]$&&&&$[-2,2]$&&&\\
		$\xb(2)$- angular velocity&&&&&&&&&\\
		$\ub(1)$- torque&$L=g$, $m=\frac{8}{g^2}$&&&&&&&&\\
		&$\tau=0.05$, $W=\set{\textbf 0}$&&&&&&&&\\
		
		\midrule
		\multirow{1}{*}{TORA}& \multirow{4}{*}{$\begin{bmatrix}\dot\xb(1)\\\dot\xb(2)\\\dot\xb(3)\\\dot\xb(4)\end{bmatrix}\in\begin{bmatrix}\xb(2)\\-\xb(1)+0.1\sin(\xb(3))\\\xb(4)\\\ub(1)\end{bmatrix}+W$}&$[-2,2]^4$&$[-1,1]$&\multirow{4}{*}{$\begin{bmatrix}0.2\\0.2\\0.2\\0.2\end{bmatrix}$}&\multirow{1}{*}{$\begin{bmatrix}0.3\end{bmatrix}$}&$[-2,2]^4$&$[-1,1]$&\multirow{4}{*}{$\begin{bmatrix}0.1\\0.1\\0.1\\0.1\end{bmatrix}$}&\multirow{1}{*}{$\begin{bmatrix}0.3\end{bmatrix}$}\\
		$\xb(1)$&&&&&&&&&\\
		$\xb(2)$&&&&&&&&&\\
		$\xb(3)$- angle&&&&&&&&&\\
		$\xb(4)$- angular velocity&&&&&&&&&\\
		$\ub(1)$- torque&$\tau=0.5$, $W=\set{\textbf 0}$&&&&&&&&\\
		
		\bottomrule
	\end{tabular}
\end{table*}
\begin{table*}
	\tiny
	\centering
	\caption{The results of regression-based controller synthesis  for finite abstractions.\textmd{ $\Xb\times\Ub$ indicates the number of discrete state-input pairs, $\eFb$, $\eBb$ denote the soundness errors, respectively, for the forward and backward representations, computed using Eq.~\eqref{eq:FW_soundness_err}, $d_F$ and $d_B$ give the graph mismatch rates for the forward and backward dynamics using using Eq.~\eqref{eq:mismatch_rate_FW_BW}, $\M_T$ gives the memory needed to store the original transition system in kB, $\M_F+\M_B$ denotes the memory taken by the representing neural networks for the forward and backward dynamics in kB, $\mathcal T_c$ denotes the total execution time for computing the compressed representations in minutes. and $\mathcal T_s$ denotes the total execution time for synthesizing the controller in minutes.}}
	\label{tab:discrete_results}
	\renewcommand{\arraystretch}{1.2}
	\setlength{\tabcolsep}{0.2em}
		\begin{tabular}{l|ccccccccc}
			\toprule
			Case study&$|\Xb|\times|\Ub|$&$\eFb$&$\eBb$&$d_F$&$d_B$&$\M_T$ (kB)&$\M_F+\M_B$  (kB)&$\mathcal T_c$ (min)&$\mathcal T_s$ (min)\\
			\midrule
			
			\multirow{1}{*}{\rotatebox{0}{2D car}} &$810000$&$\begin{bmatrix}1.02\times 10 ^{-2}\\ 1.58\times 10 ^{-2}\end{bmatrix}$& $\begin{bmatrix}2.81\times 10 ^{-2}\\ 1.17\times 10 ^{-2}\end{bmatrix}$&$6.81\times 10^{-1}$&$9.64\times 10^{-1}$&$7.76\times 10^4$&$488$&$68.58$&$8.55$\\
			
			\midrule
			
			\multirow{1}{*}{\rotatebox{0}{3D car}} &$451584$ & $\begin{bmatrix}2.05\times10^{-2}\\ 2.19\times10^{-2}\\ 2.26\times10^{-2}\end{bmatrix}$  & $\begin{bmatrix}2.48\times10^{-2}\\ 1.76\times10^{-2}\\ 2.32\times10^{-2}\end{bmatrix}$&$7.11\times 10^{-1}$&$7.85\times 10^{-1}$&$1.35\times 10 ^5$&$488$&$65.46$& $14.50$\\
			
			\midrule
			\multirow{1}{*}{\rotatebox{0}{4D car}} &$4967424$& $\begin{bmatrix}1.71\times10^{-2}\\ 2.40\times10^{-2}\\ 1.62\times10^{-2}\\1.96\times10^{-2}\end{bmatrix}$  & $\begin{bmatrix}2.05\times10^{-2}\\ 1.54\times10^{-2}\\ 1.35\times10^{-2}\\1.25\times10^{-2}\end{bmatrix}$&$4.24\times 10^{-1}$&$2.87\times 10^{-1}$&$5.58\times 10^6$&$488$&$446.23$&$20.55$\\
			\midrule
			\multirow{1}{*}{\rotatebox{0}{5D car}} &$30735936$& $\begin{bmatrix}1.41\times10^{-2}\\ 1.18\times10^{-2}\\ 1.97\times10^{-2}\\2.22\times10^{-2}\\1.93\times10^{-2}\end{bmatrix}$  & $\begin{bmatrix}2.11\times10^{-2}\\ 1.79\times10^{-2}\\ 1.13\times10^{-2}\\1.65\times10^{-2}\\2.45\times10^{-2}\end{bmatrix}$&$5.34\times 10^{-1}$&$4.25\times 10^{-1}$&$ 3.64\times 10^8 (OOM)$&$488$&$3025.14$&$312.15$\\
			
			\midrule
			
			\multirow{1}{*}{\rotatebox{0}{Inverted pendulum}} &$17360$&$\begin{bmatrix} 2.53\times 10 ^{-2}\\ 3.44\times 10 ^{-2}\end{bmatrix}$& $\begin{bmatrix}2.31\times 10 ^{-2}\\ 2.97\times 10 ^{-2}\end{bmatrix}$&$6.50\times 10^{-1}$&$5.61\times 10^{-1}$&$2.27\times 10^4$&$488$&$68.58$&$4.18$\\
			
			\midrule
			
			\multirow{1}{*}{\rotatebox{0}{TORA}} &$1433531$&$\begin{bmatrix} 2.53\times 10 ^{-2}\\ 2.67\times 10 ^{-2}\\2.39\times 10 ^{-2}\\2.24\times 10 ^{-2}\end{bmatrix}$& $\begin{bmatrix}2.21\times 10 ^{-2}\\ 2.57\times 10 ^{-2}\\ 1.88\times 10 ^{-2}\\ 3.03\times 10 ^{-2}\end{bmatrix}$&$4.34\times 10^{-1}$&$4.15\times 10^{-1}$&$1.57\times 10^7$&$488$&$241.48$&$166.16$\\
			\bottomrule
	\end{tabular}
\end{table*}

\begin{table*}
	\tiny
	\centering
	\caption{The results of classifier-based controller synthesis for finite abstractions.\textmd{ $\Xb\times\Ub$ indicates the number of discrete state-input pairs, $err_F$, $err_B$ denote the soundness errors, respectively, for the forward and backward representations, computed using Eq.~\eqref{eq:cls_soundness_errors}, $d_F$ and $d_B$ give the graph mismatch rates for the forward and backward dynamics, $\M_T$ gives the memory needed to store the original transition system in kB, $\M_F+\M_B$ denotes the memory taken by the representing neural networks for the forward and backward dynamics in kB, $\mathcal T_c$ denotes the total execution time for computing the compressed representations in minutes. and $\mathcal T_s$ denotes the total execution time for synthesizing the controller in minutes.}}
	\label{tab:discrete_results_cls}
	\renewcommand{\arraystretch}{1.2}
	\setlength{\tabcolsep}{0.2em}
		\begin{tabular}{l|ccccccccc}
			\toprule
			Case study&$|\Xb|\times|\Ub|$&$err_F$&$err_B$&$d_F$&$d_B$&$\M_T$ (kB)&$\M_F+\M_B$ (kB)&$\mathcal T_c$ (min)&$\mathcal T_s$ (min)\\
			\midrule
			
			\multirow{1}{*}{\rotatebox{0}{2D car}} &$810000$&$2.75\times 10^{-2}$& $3.27\times 10^{-2}$&$2.65\times 10^{-2}$&$2.93\times 10^{-2}$&$7.76\times 10^4$&$1.33\times 10^4$&$68.58$&$10.71$\\
			
			\midrule
			
			\multirow{1}{*}{\rotatebox{0}{3D car}} &$451584$ & $2.71\times 10^{-4}$  & $2.21\times 10^{-6}$&$3.71\times 10^{-5}$&$9.47\times 10^{-7}$&$1.35\times 10 ^5$&$1.91\times 10 ^4$&$50.74$& $12.11$\\
			
			\midrule
			\multirow{1}{*}{\rotatebox{0}{4D car}} &$4967424$& $6.24\times 10^{-4}$  & $0$&$2.84\times 10^{-4}$&$0$&$5.58\times 10^6$&$2.37\times 10^4$&$565.13$&$24.58$\\
			
			\midrule
			\multirow{1}{*}{\rotatebox{0}{5D car}} &$30735936$& $3.41\times 10^{-5}$  & $5.33\times 10^{-8}$&$3.21\times 10^{-5}$&$2.19\times 10^{-8}$&$3.64\times 10^8 (OOM)$&$3.27\times 10^4$&$3421.21$&$215.88$\\
			
			\midrule
			\multirow{1}{*}{\rotatebox{0}{Inverted pendulum}} &$17360$& $6.03\times 10^{-2}$  & $5.85\times 10^{-2}$&$0$&$0$&$2.27\times 10^4$&$2.08\times 10^4$&$8.21$&$8.33$\\
			
			\midrule
			\multirow{1}{*}{\rotatebox{0}{TORA}} &$1433531$& $1.27\times 10^{-1}$ & $1.26\times 10^{-1}$&$1.55\times 10^{-1}$&$1.48\times 10^{-1}$&$1.57\times 10^7$&$2.38\times10^4$&$234.87$&$159.75$\\
			\bottomrule
	\end{tabular}
\end{table*}
We evaluate the performance of our proposed algorithms 
on 
several control systems. Dynamics of our control systems are listed in Tab.~\ref{tab:case_studies}. 
We used configurations (1) and (2) in Tab.~\ref{tab:case_studies}, respectively, for evaluating our methods for synthesis and deployment. 
We construct the transition system in all the case studies using the sampling approach in \cite{Milad:2022}. This approach generates $T_F$ using sampled trajectories while providing confidence on the correctness of $T_F$.
Our experiments were performed on a cluster with Intel Xeon E7-8857 v2 CPUs (32 cores in total) at 3GHz, with 100GB of RAM. 
For training neural networks, we did not use a distributed implementation as we found that distributing the process across GPUs actually decelerates the process. However, for the rest of our compression and synthesis algorithms, we used a distributed implementation.

\smallskip
\noindent\textbf{Synthesis.}\
We considered the $\ell_\infty$ ball centered at $(4,4)$ with the radius $0.8$ over the Euclidean plane as the target set for the multi-dimensional car examples, $[-0.5,0.5]\times[-1,1]$ for the inverted pendulum example, and $[-1,1]^4$ for the TORA example. 
To evaluate our 
corrected neural method described in Subsec.~\ref{subsec:TS_compression_regression}, 
we set the list of neuron numbers in different layers as $(n+m,20,40,30,2n)$, select the activation functions to be hyperbolic tangent, and set the learning rate to be $\lambda=0.001$. As discussed in 
Subsec.~\ref{sec:TS_compression_classification}, the corrected neural representations for finite abstractions can also be constructed by solving a classification problem. 
To evaluate this method, we set the list of neuron numbers in different layers for both $\NN_F$ and $\NN_B$ as $(n+m,40,160,160,160,160,160,160,160,160,500,800,2\sum_{i=1}^n|\Xb(i)|)$, select the activation functions to be ReLU, and set the learning rate to be $\lambda=0.0001$. 
 We used stochastic gradient descent method with the corresponding learning rate for training the neural networks \cite{Ruder:2016}. Tabs.~\ref{tab:discrete_results} and \ref{tab:discrete_results_cls} illustrate the synthesis results related to our experiments for finite abstractions, using the regression-based and classification-based methods, respectively. 
 Although we used the same neural network structure for all the examples, soundness errors take small values that are bounded by $3.44\times10^{-2}$ as the maximum of $\eFb$ and $\eBb$ in the regression-based method, and by 
 $1.27\times 10^{-1}$ as the maximum of $err_F$ and $err_B$ in the classification-based method. Moreover, memory requirement of our proposed regression-based and classification-based methods at higher dimensions remains almost constant while the size of the transition system increases exponentially (see the illustration shown in Fig.~\ref{fig:results} (Left) for the multi-dimensional car case studies). 
Further, we notice that the regression-based method results in higher mismatch rates $d_F$ and $d_B$ compared to the classification-based method: on average, $5.87\times 10^{-1}$ versus $3.03\times 10^{-2}$ for $d_F$, and $6.15\times 10^{-1}$ versus $2.96\times 10^{-2}$ for $d_B$ (see the illustration shown in Fig.~\ref{fig:results} (Right) for the multi-dimensional car case studies).
Therefore, using the classification-based method, while being sound, produces a smaller graph, which is less restrictive for the synthesis purpose. 
Most importantly, memory requirement using both our approaches is way less  than the memory needed to store the original (forward) transition system ($\M_F+\M_B<<\M_T$).
Regression-based method reduces the memory requirements by a factor of $1.31\times 10^5$ and up to $7.54\times 10^5$. However, the classification-based method reduces the memory requirements by 
a factor of $2.01\times 10^3$ and up to $1.12\times 10^4$. 
This shows that the regression-based method requires less memory compared to the classification-based method.
\begin{figure}
	\centering
	\begin{tabular}{ll}
		\includegraphics[scale=.3]{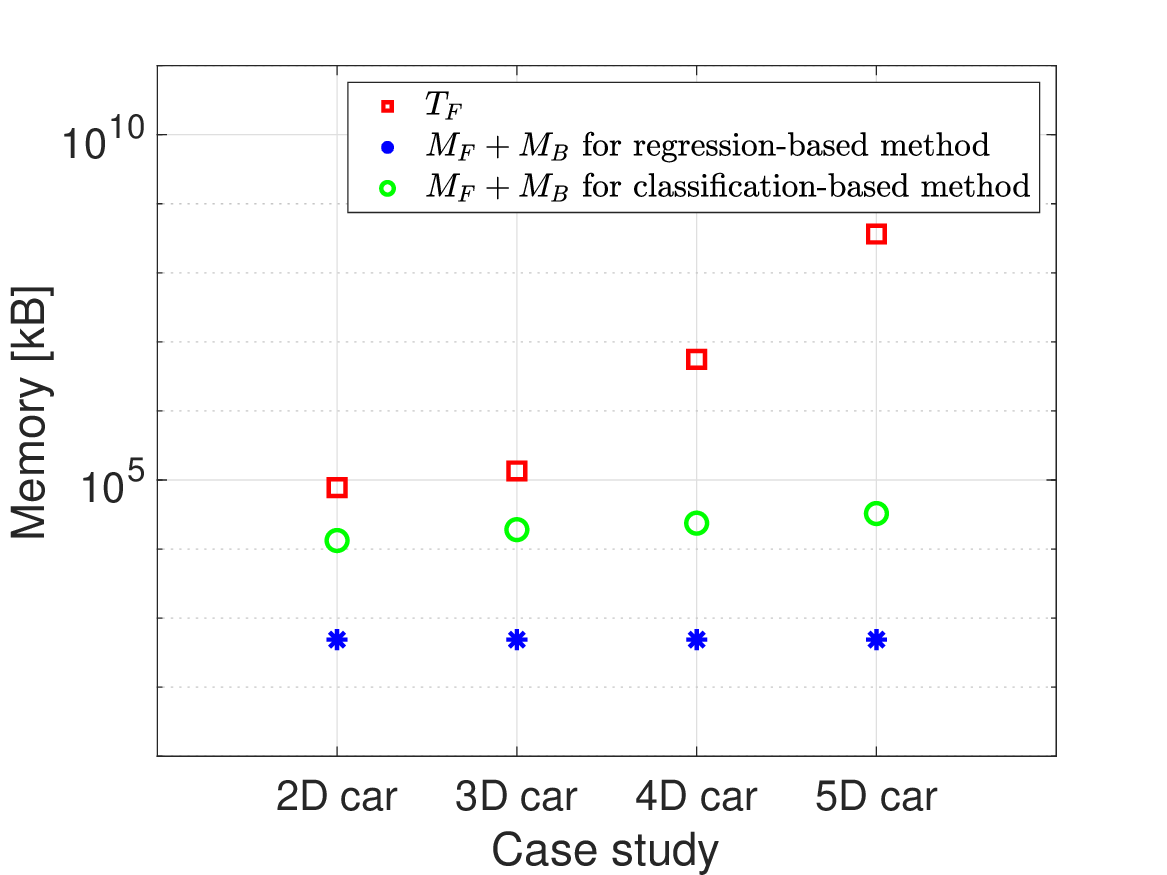}
		&
		\includegraphics[scale=.3]{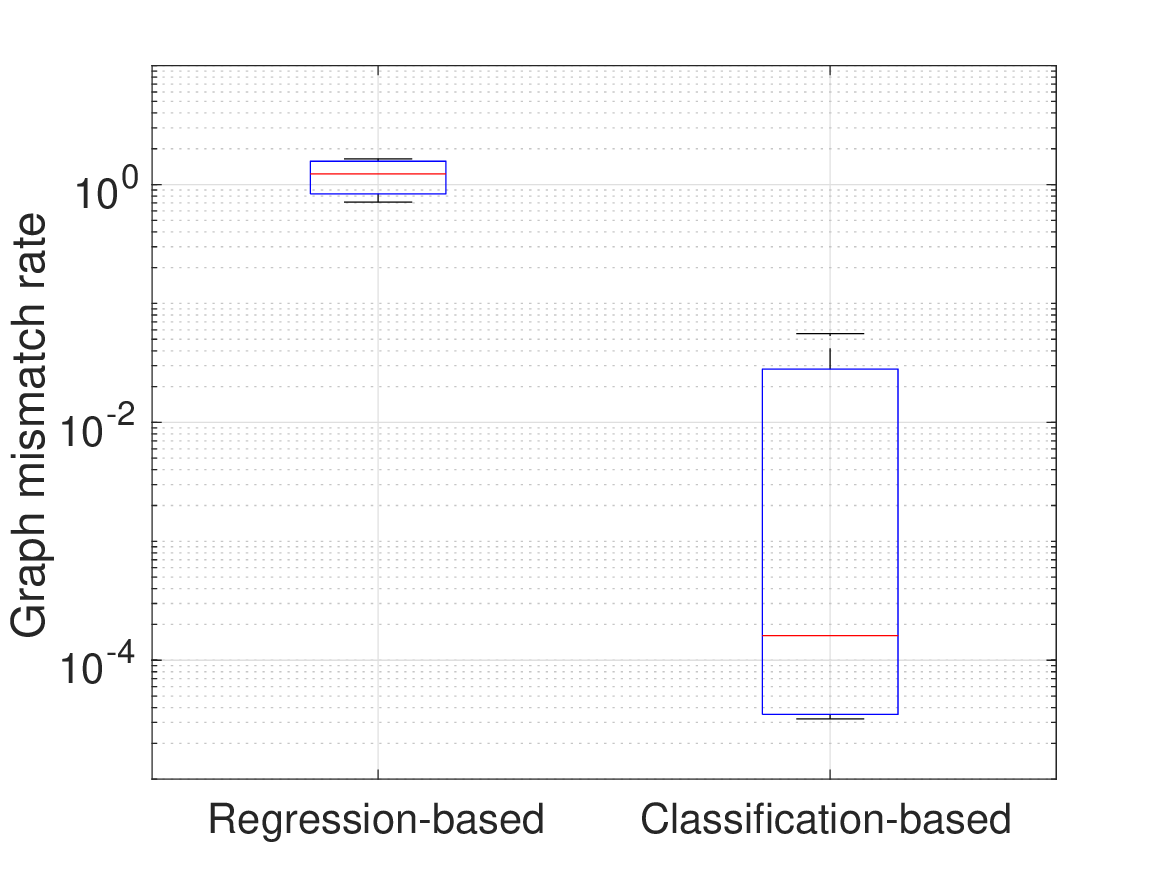}
	\end{tabular}
	\caption{Left: Memory requirement of different methods for storing transition systems of multi-dimensional cars (cf. Tab.~\ref{tab:case_studies}) in logarithmic scale. Right: Distribution of total graph mismatch rate ($d_F+d_B$) for our proposed methods in logarithmic scale.
	}
	\label{fig:results}
	\vspace{-.3cm}
\end{figure}

\smallskip
\noindent\textbf{Deployment.}\
Tab.~\ref{tab:controller_results} lists our experimental results for compressing the symbolic controllers. For $\NN_C$, 
we set the list of neuron numbers in different layers for both $\NN_F$ and $\NN_B$ as $(n,20,80,80,80,80,80,160,|\Ub|)$, select the activation functions to be rectified linear unit (ReLU), and set the learning rate to be $\lambda=0.0001$. It can be noticed that $err_C$ is very small for all the examples. Therefore, we only need to store a very small portion of $C$ in addition to $\NN_C$.
As it can be observed in Tab.~\ref{tab:controller_results}, our method has been successful in computing representations which are very accurate and compact-in-size ($\M_{\hat C}<<\M_C$).

\begin{table*}
	\tiny
	\centering
	\caption{The results of controller compression.\textmd{
			$|C|$ gives the number of state-input pairs in the original controller, $err_C$ denotes the portion of the states at which the representing neural network produces non-valid control inputs computed using Eq.~\eqref{eq:err_C}, $\M_C$ gives the memory needed to store the original controller in kB, $\M_{\hat C}$ denotes the memory taken by the representing neural network in kB, and $\mathcal T$ denotes the total execution time for our implementation in minutes.}}
	\label{tab:controller_results}
	\renewcommand{\arraystretch}{1.2}
	\setlength{\tabcolsep}{0.2em}
		\begin{tabular}{l|ccccc}
			\toprule
			Case study&$|C|$&$err_C$&$\M_C$ (kB)&$\M_{\hat C} (kB)$&$\mathcal T$ (min)\\
			\midrule
			
			\multirow{1}{*}{\rotatebox{0}{2D car}} &$2.15
			\times 10 ^ 6$&$1.85\times 10^{-5}$& $2.75\times 10 ^5$
			&$1.21\times 10 ^3$&$6.31$\\
			
			\midrule
			
			\multirow{1}{*}{\rotatebox{0}{3D car}} &$2.87\times 10 ^6$&$2.16\times 10^{-3}$& $4.65\times 10 ^5$
			&$1.05\times 10 ^3$&$19.14$\\
			
			\midrule
			\multirow{1}{*}{\rotatebox{0}{4D car}}  &$9.35 \times 10^7$&$3.63\times10^{-2}$& $2.24\times 10^6$
			&$1.35\times 10^3$&$39.48$\\
			
			\midrule
			\multirow{1}{*}{\rotatebox{0}{5D car}}  &$1.69 \times 10^9$&$4.51\times10^{-3}$& $4.71\times 10^7$
			&$1.48\times 10^3$&$201.86$\\
			
			\midrule
			\multirow{1}{*}{\rotatebox{0}{Inverted pendulum}}  &$8.16 \times 10^5$&$1.08\times10^{-3}$& $7.83\times 10^4$
			&$8.92\times 10^2$&$7.51$\\

			\midrule
			\multirow{1}{*}{\rotatebox{0}{TORA}}  &$4.78 \times 10^7$&$3.78\times10^{-4}$& $7.65\times 10^6$
			&$8.92\times 10^2$&$113.97$\\
			\bottomrule
	\end{tabular}
\end{table*}



\noindent\textbf{Parametrization.} Our approach requires selecting the hyperparameters of the training process and choosing the structure of the neural networks. We have performed several experiments to select the hyperparameters of the training (e.g., the learning rate, epoch number, and batch size). Regarding the structure of the neural networks, we have explored different choices such as the type of the activation functions (hyperbolic tangent, ReLU, etc.), number of neurons per layer, and the depth. Increasing the complexity of the neural network, by increasing the number of neurons per layer or depth, leads to a better performance.
Note that the neural networks employed in our setting are not supposed to make any generalization over unseen data. Therefore, our approach does not suffer from over-parametrization of the neural networks. We have demonstrated this in Fig.~\ref{fig:parametrization} by providing the error as a function of the depth of the neural representation for the 3D car example. The error always decreases by increasing the depth of the neural representation. Therefore, the structure of the neural representations can be selected for having an acceptable accuracy within a given time bound for the training process.
\begin{figure}
	\centering
	\begin{tabular}{lll}
		\includegraphics[scale=.23]{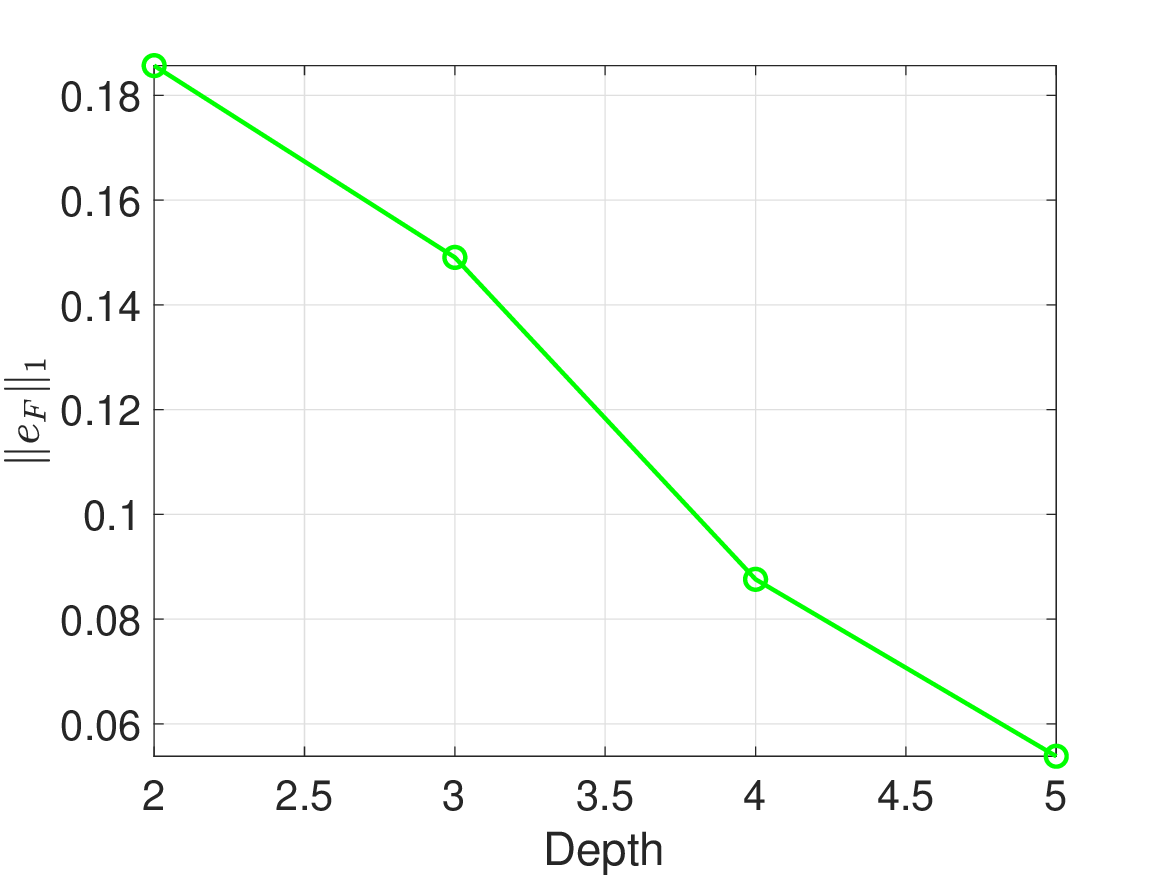}
		&
		\includegraphics[scale=.23]{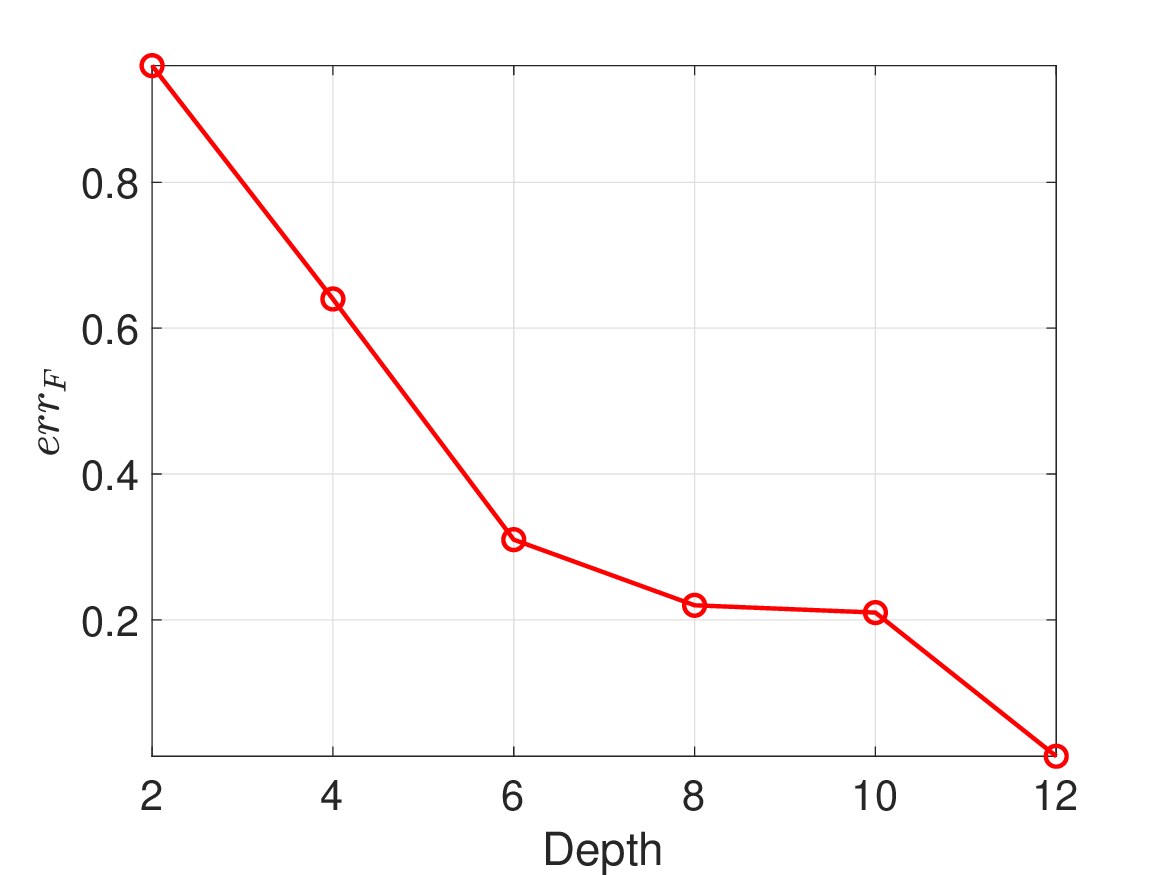}&\includegraphics[scale=.23]{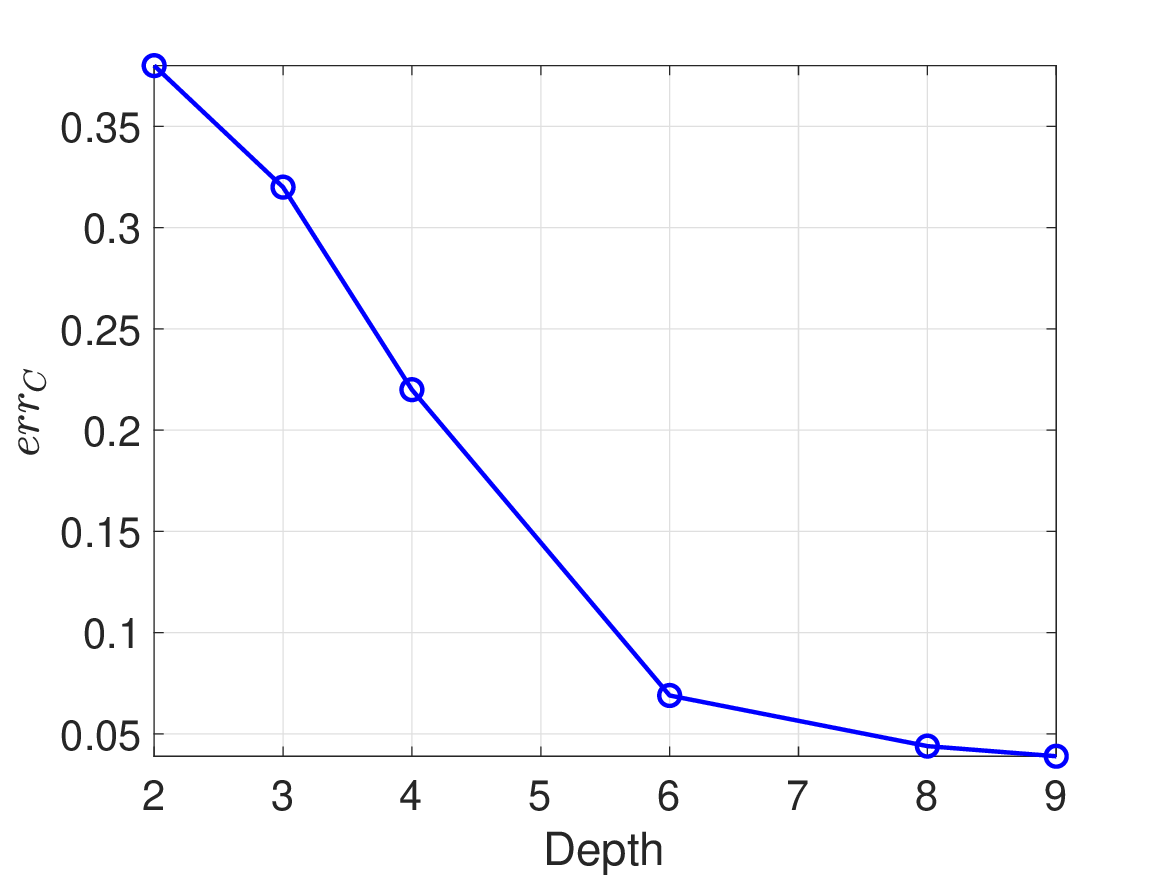}
	\end{tabular}
	\caption{Demonstrating the effect of increasing the depth of the neural representation on the norm of the soundness error $\eFb$ (cf. Eq.~\eqref{eq:FW_soundness_err}) for regression-based controller synthesis (Left), the soundness error $err_F$ (cf. Eq.~\eqref{eq:cls_soundness_errors}) for classification-based controller synthesis (Middle), and the misclassification rate (cf. Eq.~\eqref{eq:err_C}) for deployment (Right). The experiments are performed on the 3D car example.}
	\label{fig:parametrization}
	\vspace{-.3cm}
\end{figure}

%% file: conclusion.tex
\section{Discussion and Conclusions}\label{sec:conclusion}
In this paper, we
considered abstraction-based methods for controller synthesis to satisfy high-level temporal requirements. We addressed the (exponentially large) memory requirements of these methods in both synthesis and deployment.
Using the expressive power of neural networks,
we proposed memory-efficient methods to compute compressed representations for both forward and backward dynamics of the system. With focus on reach-avoid specifications, we showed how to perform synthesis using corrected neural representations of the system.
%
 We also proposed a novel formulation for finding compact corrected neural representations of the controller to reduce the memory requirements in deploying the controller. Finally, we evaluated our approach on multiple case studies, showing reduction in memory requirements by orders of magnitude, while providing formal correctness guarantees. 
 
\smallskip
\noindent\textbf{Extension to more general specifications.}\
Our approach is based on computing an under-approximation of $\Pre$ and over-approximation of $\Post$ operators. Therefore, it can be applied to any synthesis problem whose solution is characterized based on these operators. This means our approach can be applied to control synthesis for other linear temporal logic specifications including safety, B\"uchi, and Rabin objectives.

\smallskip
\noindent\textbf{Reusability of the computed representations.}\
 Our approach computes the corrected neural representations that is sound on the whole state space. These representations can be used for any other problem defined over the same finite abstraction.

\smallskip
\noindent\textbf{Application to systems with known analytical model.}\
Our approach is efficient in providing compact representations for a given finite abstraction at the cost of increasing the off-line computational time. This is regardless of constructing the finite abstraction using model-based methods or (correct) data-driven methods. Model-based on-the-fly synthesis methods will utilize numerical solutions of differential equations when the analytical model of the system is known with available bounds on the continuity properties of the system. These methods may perform better in case solving the corresponding differential equations is faster than making a forward pass through the neural representation.

\noindent
\textbf{Comparison with a baseline method.}
We have demonstrated the effectiveness of our method on a number of case studies in compressing finite transition systems and controllers which are stored in the form of look-up tables. In the introduction and related work sections of our paper, we have discussed why other methods cannot be used to solve our problem. In below, we have listed our main arguments.
\begin{itemize}
	\item While transition systems and controllers can be encoded using BDDs instead of look-up tables, the memory blow-up problem still exists for systems of higher dimensions. However, using our technique, we empirically show that the size of the computed representations is not necessarily affected by size of the original mapping. See for example Fig.~\ref{fig:results} (Left), wherein the memory required by the trained compressed representation stays at 488 kB, despite the fact that the required memory by the original transition system has increased by a factor of 5000.
	\item Also, our synthesis setting is different from the one considered in references \cite{Reissig:2022,Girard:2022,Rungger:2022}, wherein memory-efficient synthesis methods are proposed based on a (compact) analytical description of the nominal dynamics of the system and its growth bound. We consider the case wherein the input is a huge finite transition system which can also be learned from simulations. 
	\item Finally, while the control determinization and compression schemes proposed in \cite{Girard:2022, ZAPREEV:2018} are based on the BDD and ADD encodings of the controller, the only methodologically that is in a similar spirit as our deployment approach is the symbolic regression of \cite{ZAPREEV:2018}. As mentioned by the authors of \cite{ZAPREEV:2018}, their regression-based method is not able to represent the original controller with an acceptable accuracy. Our superior performance is mainly because of our classification-based formulation, as opposed to a regression-based formulation.
\end{itemize}

\smallskip
\noindent\textbf{Utilizing invertible neural networks.} Our method requires training two different neural networks associated with the forward and backward dynamics. A possible future research direction would be to use invertible neural networks instead of training two separate neural networks.
However, given the specific application and inherent differences between our approach and the successful experiences with invertible neural networks, it is currently not obvious to us that the same performance would be accessible.

\smallskip
\noindent\textbf{Choice of hyper-rectangular reachable sets.}
Hyper-rectangular sets are the most popular choice for representing reachable sets in abstraction-based controller synthesis (see e.g., \cite{reissig2016feedback}). Although other templates could also be used to represent reachable sets, the conservativeness of the over-approximation reduces as the size of the discretization decreases. Our approach can in principle be applied to any parametrization of the reachable sets.

\smallskip
\noindent\textbf{Robustness against adversarial examples.}
In general, it is true that neural networks are not robust to adversarial examples, meaning that a small change in their input could give large enough changes in their output resulting in errors. We emphasize that going from a look-up table representation of the controller to a neural network representation does not influence the robustness against adversarial attacks unless the attacker gains access to middle layers of the neural network. Any adversarial attack on the inputs of the neural network can be studied using similar techniques and concepts from the robustness analysis of abstraction-based methods and is independent of the controller representation.

